\documentclass[12pt]{article}
\usepackage{amsmath,bbm,color,xcolor}
\usepackage{amssymb}
\usepackage{dsfont}
\usepackage{float}
\usepackage{graphicx}
\usepackage{caption}
\usepackage{subcaption}
\usepackage[margin=1in]{geometry}
\usepackage[noend]{algpseudocode}
\usepackage{authblk}
\usepackage[ruled, vlined, linesnumbered]{algorithm2e}
\usepackage{amsthm}
\usepackage{pdfpages}
\usepackage[colorlinks,citecolor=blue,urlcolor=blue,breaklinks]{hyperref}
\usepackage[round]{natbib}

\DeclareMathOperator*{\argmax}{argmax}

\DeclareMathOperator{\diag}{diag}
\DeclareMathOperator{\rank}{rank}
\newtheorem{theorem}{Theorem}
\newtheorem{lemma}{Lemma}
\newtheorem{proposition}{Proposition}

\title{High-dimensional changepoint estimation with heterogeneous missingness}

\author{Bertille Follain$^{*}$, Tengyao Wang$^{\dagger}$ and Richard J. Samworth$^*$ \\
$^*$Statistical Laboratory, University of Cambridge \\
$^\dagger$Department of Statistical Science, University College London 
}
\begin{document}
\maketitle

\begin{abstract}
We propose a new method for changepoint estimation in partially-observed, high-dimensional time series that undergo a simultaneous change in mean in a sparse subset of coordinates.  Our first methodological contribution is to introduce a `MissCUSUM' transformation (a generalisation of the popular Cumulative Sum statistics), that captures the interaction between the signal strength and the level of missingness in each coordinate.  In order to borrow strength across the coordinates, we propose to project these MissCUSUM statistics along a direction found as the solution to a penalised optimisation problem tailored to the specific sparsity structure.  The changepoint can then be estimated as the location of the peak of the absolute value of the projected univariate series.  In a model that allows different missingness probabilities in different component series, we identify that the key interaction between the missingness and the signal is a weighted sum of squares of the signal change in each coordinate, with weights given by the observation probabilities.   More specifically, we prove that the angle between the estimated and oracle projection directions, as well as the changepoint location error, are controlled with high probability by the sum of two terms, both involving this weighted sum of squares, and  representing the error incurred due to noise and the error due to missingness respectively.  A lower bound confirms that our changepoint estimator, which we call \texttt{MissInspect}, is optimal up to a logarithmic factor.  The striking effectiveness of the \texttt{MissInspect} methodology is further demonstrated both on simulated data, and on an oceanographic data set covering the Neogene period.  
\end{abstract}

\textit{Key words}: changepoint estimation; missing data; high-dimensional data; segmentation; sparsity

\section{Introduction}
The Big Data era offers the exciting prospect of being able to transform our understanding of many scientific phenomena, but at the same time many traditional statistical techniques may perform poorly, or may no longer be computable at all, when applied to contemporary data challenges.  A core assumption that underpins much of statistical theory, as well as the way in which we think about statistical modelling, is that our data are realisations of independent and identically distributed random variables.  However, practical experience reveals that this is typically unrealistic for modern data sets, and developing methods and theory to handle departures from this important but limited setting represents a key theme for the field.

In contexts where data are collected over time, one of the simplest generalisations of an independent and identically distributed data stream is given by \emph{changepoint} models.  Here, we postulate that our data may be segmented into shorter, homogeneous series.  Of course, the structural break, or changepoint, between these series is often of interest in applications, such as distributed denial of service monitoring of network traffic \citep{peng2004proactively}, disease progression tracking via the alignment of electronic medical records \citep{huopaniemi2014disease} and the analysis of `shocks' in stock price data \citep{chen1997testing}.   

Another issue that turns out to be critical in working with Big Data in practice is that of missing data.  One reason for this is that when each observation is high-dimensional, it is frequently the case that most or even every observation has missingness in some coordinates; thus a complete-case analysis, which simply discards such observations, is unviable \citep{zhu2019high}.   

The aim of this paper is to study the core, high-dimensional changepoint problem of a sparse change in mean, but where our data are corrupted by missingness.  In fact, in cases where our data arise as discrete observations of several continuous processes, the observation times in different coordinates may not be the same, and such a setting also fits within our framework.  A key feature of both our methodology and theory is that we wish to be able to handle \emph{heterogeneous} missingness, i.e.~where the levels of missingness may differ across coordinates.  Specifically, our primary theoretical goal is to understand the way in which the missingness \emph{interacts} with the signal strengths in the different series to determine the difficulty of the problem.

In Section~\ref{Sec:Methodology}, we consider a setting where the practitioner has access to a partially-observed $p \times n$ data matrix, where $p$ is the number of series (coordinates) being monitored, and $n$ is the number of time points.  We seek to identify a time at which the $p$-dimensional mean vector changes, in at least one coordinate.  One of the key ideas that underpins our methodological contribution is to define a new version of the popular Cumulative Sum (CUSUM) transformation \citep{Page1955} that is able to handle the missingness appropriately.  This operation, which we refer to as the MissCUSUM transformation, returns a $p \times (n-1)$ matrix, and the intuition is that in coordinates that undergo a change in mean, the transformed series should peak in absolute value near the changepoint.  One of the main advantages of our proposal is that it avoids the need to impute missing data\footnote{In fact, our initial approach to this problem was to consider iterating between the imputation of missing entries using row means on either side of a putative changepoint, and then updating the current changepoint location estimate using the imputed data matrix.  This turned out to perform poorly, because the imputation step tended to reinforce bias in the changepoint estimate, leading to the iterations becoming stuck (potentially far from the true location) very quickly.}.  

Since the changepoint location is shared across the signal coordinates, it is natural to seek to borrow strength across the different data streams to estimate the changepoint.  To this end, our next goal is to estimate a projection direction, in order to convert the MissCUSUM transformation into a univariate CUSUM series.  Such a projection direction should ideally maximise the signal-to-noise ratio of the projected series.  When the data are fully observed, the oracle projection direction turns out to be the leading left singular vector of the (rank one) CUSUM transformation of the mean matrix.  This facilitates estimation approaches based on entrywise $\ell_1$-penalised $M$-estimation, as in the \texttt{inspect} algorithm of \citet{wang2016highdimensional}.  A crucial difference when we have to handle missing data, however, is that the MissCUSUM transformation of the mean matrix is no longer of rank one, which means that the entrywise $\ell_1$-penalty no longer adequately captures the sparsity structure of the vector of mean change.  Instead, we introduce a new optimisation problem that penalises the $\ell_1$-norm of the leading left singular vector of a rank one approximation of the MissCUSUM transformation.  This methodological proposal, which we call \texttt{MissInspect}, leads to considerably improved performance.  Implementation code for our method is available in the GitHub repository \url{https://github.com/wangtengyao/MissInspect}.

A further benefit of the \texttt{MissInspect} methodology is that it is amenable to theoretical analysis.  In particular, we study a Missing Completely At Random model with \emph{row homogeneous} missingness; in other words, the observation probability remains constant in each row, but may vary arbitrarily across rows.  In Proposition~\ref{Prop:SineAngle} in Section~\ref{Sec:Theory}, we provide a high-probability bound on the angle between the estimated and oracle projection directions.  Theorems~\ref{Thm:SlowRate} and~\ref{Thm:FastRate} then establish high-probability bounds on the accuracy of the estimated changepoint location whenever the estimated and oracle projection directions are sufficiently well aligned, for a sample splitting variant of our algorithm.  Theorem~\ref{Thm:SlowRate} provides a very general guarantee, while Theorem~\ref{Thm:FastRate} establishes a faster rate whenever the observation probability in each row satisfies a lower bound.  This faster rate comprises two terms, representing the error incurred due to noise in the observations, and the error due to missingness, respectively.  The key quantity in both of these terms turns out to be a weighted Euclidean norm of the vector of mean change, where the weights are given by the observation probabilities in each row.  This weighted average therefore captures the interaction between the signal strength and the missingness probabilities, and suggests that our analysis handles effectively the heterogeneity of the missingness across rows (a more naive analysis would see the worst-case observation probability appearing in the bounds).  This intuition is confirmed by our minimax lower bound (Theorem~\ref{Thm:LowerBound}), which indicates that the \texttt{MissInspect} algorithm attains the minimax rate of convergence in all problem parameters, up to a logarithmic factor. 

Section~\ref{Sec:Numerical} explores the empirical performance of our \texttt{MissInspect} methodology.  We study the ability of the algorithm to estimate both the oracle projection direction and the changepoint location, and compare with an alternative algorithm that imputes the missing entries using the well-known \texttt{softImpute} algorithm \citep{mazumder2010spectral}, and then runs the original \texttt{inspect} algorithm.  We find that the \texttt{MissInspect} algorithm considerably outperforms this approach, and provides further evidence of its practical utility in missing data settings.  In this section, we also present an application of the \texttt{MissInspect} methodology to detect changes in oceanographic currents from carbon isotope measurements extracted from cores drilled into the ocean floor.  Section~\ref{Sec:Extensions} discusses various methodological and theoretical extensions of our proposal to more complicated problems, such as the estimation of multiple changepoints, or more general data generating and missingness mechanisms.  Proofs or our main results are given in Section~\ref{Sec:Proofs}, with auxiliary results deferred to the Appendix.

    The study of changepoint problems dates at least back to \citet{Page1955}, and has since found applications in many different areas, including genetics \citep{olshen,Zhangetal2010}, disease outbreak watch \citep{sparks}, aerospace engineering \citep{henry} and functional magnetic resonance imaging studies \citep{aston2012}, in addition to those already mentioned.  Entry points to the literature include \citet{csorgo} and \citet{rice}.  In high-dimensional changepoint settings, where we may have a sparsity assumption on the coordinates of change, prior work includes \citet{Bai2010}, \citet{Zhangetal2010}, \citet{HorvathHuskova2012}, \citet{Cho_2014}, \citet{chan2015optimal}, \citet{Jirak2015}, \citet{Cho_2016}, \citet{soh2017high}, \citet{wang2016highdimensional}, \citet{enikeeva2013highdimensional}, \citet{padilla2019optimal} and \citet{liu2021minimax}.  The only works of which we are aware on changepoint estimation with missing data are those of \citet{Xie_2013} and \citet{londschien2021change}, both of which consider different settings to ours.  \citet{Xie_2013} study a situation where partially-observed sequential data lie close to a time-varying, low-dimensional submanifold embedded within an ambient space; on the other hand, \citet{londschien2021change} consider changepoint detection in graphical models.  Finally we mention that our focus in this work is on the \emph{offline} version of the changepoint estimation problem, where the practitioner sees the whole data set prior to determining a changepoint location.  The corresponding online version, where data are observed sequentially and the challenge is to declare a change as soon as possible after it has occurred, has also received attention in recent years; see, e.g.,  \citet{mei2010efficient}, \citet{xie2013sequential}, \citet{chan2017optimal} and \citet{chen2020highdimensional}.

We conclude this section by introducing some notation that is used throughout the paper.  Given $n \in \mathbb{N}$, we let $[n] := \{1,\ldots,n\}$.  For a vector $u=(u_1, \ldots , u_M)^\top \in \mathbb{R}^M$, a matrix $A=(A_{ij}) \in \mathbb{R}^{M\times N }$ and for $r \in [1,\infty)$, we write $\|u\|_r := \bigl(\sum_{i=1}^M |u_i|^r\bigr)^{1/r} $ and $\|A\|_r := \bigl(\sum_{i\in[M]} \sum_{j\in[N]} |A_{ij}|^r\bigr)^{1/r}$ for their entrywise $\ell_r$-norms, as well as $\|u\|_\infty := \max_{i\in[M]} |u_i|$ and $\|A\|_\infty := \max_{i\in[M], j\in[N]} |A_{ij}|$.  Given $\boldsymbol{q} = (q_1,\ldots,q_M)^\top \in [0,1]^M$, we write $\sqrt{\boldsymbol{q}} := (\sqrt{q_1},\ldots,\sqrt{q_M})^\top$ and let $\|u\|_{r,\boldsymbol{q}} := \bigl(\sum_{i=1}^M |u_i|^r q_i\bigr)^{1/r}$.  Writing $\sigma_1(A), \ldots , \sigma_s(A)$ for the non-zero singular values of $A$, where $s := \mathrm{rank}(A)$, we let $\|A\|_{\mathrm{op}} := \max_{i \in [s]} \sigma_i(A)$, $\|A\|_* := \sum_{i=1}^{s} \sigma_i(A)$ and $\|A\|_{\mathrm{F}}:= \|A\|_2=\{\sum_{i=1}^{s}\sigma_i(A)^2\}^{1/2}$  denote its operator, nuclear and Frobenius norms respectively. We also write $\|u\|_0 := \sum_{i=1}^M \mathbbm{1}_{\{u_i \neq 0\}}$.  We denote by $\diag(u)$ the $M \times M$ diagonal matrix with $u$ as its diagonal.  For $S \subseteq
[M]$ and $T \subseteq [N]$, we write $u_S := (u_i : i \in S)^\top \in \mathbb{R}^{|S|}$ and write $M_{S,T} \in \mathbb{R}^{|S| \times |T|}$ for the sub-matrix of $A$ obtained by extracting the rows and columns with indices in $S$ and $T$ respectively. For two matrices $A,B \in  \mathbb{R}^{M \times N}$, we denote their trace inner product as $\langle A,B \rangle := \mathrm{tr}(A^\top B)$. We also denote their Hadamard product as $A \circ B \in \mathbb{R}^{M \times N}$. For non-zero vectors $u, v \in \mathbb{R}^M$, we write
$$ \angle(u,v) := \cos^{-1}\left( \frac{|\langle u,v\rangle|}{\|u\|_2\|v\|_2} \right) $$ 
for the acute angle bounded between them. We let $\mathbb{B}^M := \{\mathbb{R}^M : \|x\|_2 \leq 1\}$ and $\mathbb{S}^{M-1} :=\{x \in \mathbb{R}^M : \|x\|_2 =1\}$ denote the unit Euclidean ball and sphere in $\mathbb{R}^M$ respectively, and define $\mathbb{S}^{M-1}(k) := \{x \in \mathbb{S}^{M-1} : \|x\|_0  \leq k \}$.  Given positive sequences $(a_n), (b_n)$, we write $a_n \lesssim b_n$ to mean that there exists a universal constant $C > 0$ such that $a_n \leq Cb_n$ for all $n$. 

\section{\texttt{MissInspect} methodology}
\label{Sec:Methodology}

Throughout this work, we will assume that the practitioner has access to a partially-observed $p \times n$ data matrix.  We will denote the full data matrix as $X = (X_{j,t}) \in \mathbb{R}^{p \times n}$, and let $\Omega = (\omega_{j,t}) \in \{0,1\}^{p \times n}$ denote the \emph{revelation matrix}, so that $\omega_{j,t} = 1$ if $X_{j,t}$ is observed, and is equal to zero otherwise.  Formally, then, we can regard the observed data as $(X \circ \Omega,\Omega)$; note here, that since the practitioner has access to the matrix $\Omega$, they are able to distinguish between an observed zero and a zero caused by missingness in $X \circ \Omega$.  For our theoretical analysis, the $\circ$ notation is a convenient way of avoiding the need to introduce an `\texttt{NA}' category for missing values. 

In our theory, we will regard $X$ as a realisation of a random matrix, whose mean matrix we denote by $\boldsymbol{\mu} = (\mu_1,\ldots,\mu_n) \in \mathbb{R}^{p \times n}$.  The changepoint structure of $\mu$ is encoded via the assumption that there exist $z \in [n-1]$ and $\mu^{(1)}, \mu^{(2)} \in \mathbb{R}^p$ with $\theta := \mu^{(2)} - \mu^{(1)} \neq 0$ such that
\begin{equation}
  \label{Eq:SingleChangepoint}
    \mu_1 = \cdots =\mu_z = \mu^{(1)} \ \text{and} \ \mu_{z+1} = \cdots = \mu_n = \mu^{(2)}.
\end{equation}
In Section~\ref{Sec:Theory}, we will assume that the change in mean is sparse, in the sense that $\|\theta\|_0 \leq k$ for some $k$ that is typically much smaller than $p$.  However, we remark that our methodology is adaptive to this unknown sparsity level.  

Our goal is to estimate the changepoint location $z$.  To this end, we first introduce a new version of the CUSUM transformation that is appropriate in our missing data setting.  Writing $L_{j,t} := \sum_{r=1}^t \omega_{j,r}$, $R_{j,t} := \sum_{r=n-t+1}^n \omega_{j,r}$ and $N_j:=L_{j,n} = R_{j,n}$ for $j \in [p]$ and $t \in [n]$, we define the MissCUSUM transformation $\mathcal{T}_{p,n}^{\mathrm{Miss}} : \mathbb{R}^{p \times n} \times \{0,1\}^{p \times n} \rightarrow \mathbb{R}^{p \times (n-1)}$ by
\[
    [\mathcal{T}_{p,n}^{\mathrm{Miss}}(M,\Omega)]_{j,t} := \sqrt{\frac{L_{j,t}R_{j,n-t}}{N_j}} \biggl( \frac{1}{R_{j,n-t}}\sum_{r=t+1}^n  (M \circ \Omega)_{j,r} - \frac{1}{L_{j,t}}\sum_{r=1}^t  (M \circ \Omega)_{j,r} \biggr)
\]
when $L_{j,t}>0$ and $R_{j,n-t}>0$, and define $[\mathcal{T}_{p,n}^{\mathrm{Miss}}(M,\Omega)]_{j,t} := 0$ otherwise.  Since the subscripts $p$ and $n$ of $\mathcal{T}^{\mathrm{Miss}}_{p,n}$ can be inferred from the dimensions of its arguments, we will frequently abbreviate this transformation as $\mathcal{T}^{\mathrm{Miss}}$.  We note that this transformation only depends on~$M$ through $M \circ \Omega$.  In practice, we will always apply this transformation to pairs of the form $(M \circ \Omega,\Omega)$; in other words, an entry of the first argument is zero whenever the corresponding entry of the second argument is zero.  When the data matrix is fully observed (i.e.\ $\Omega$ is an all-one matrix), the MissCUSUM transformation reduces to the standard CUSUM transformation $\mathcal{T}(M) := \mathcal{T}^{\mathrm{Miss}}(M,\Omega)$.

A key feature of the MissCUSUM transformation is that it captures the interaction between the signal strength and the number of observations in each coordinate.  To illustrate this, we focus on a single  ($j$th) coordinate, and the noiseless setting where $X = \boldsymbol{\mu}$.  In this case, the peak level of the absolute MissCUSUM transformation is $|\theta_j|\sqrt{\frac{L_{j,z}R_{j,n-z}}{N_j}}$.  Since 
\[
\sqrt{\frac{\min(L_{j,z},R_{j,n-z})}2} \leq \sqrt{\frac{L_{j,z}R_{j,n-z}}{N_j}} \leq \sqrt{\min(L_{j,z},R_{j,n-z})},
\]
we see that the peak level in the $j$th coordinate is controlled by the absolute mean change~$|\theta_j|$, together with the effective sample size $\min(L_{j,z},R_{j,n-z})$. 

Another interesting property of the MissCUSUM transformation is that the multivariate setting allows us to borrow strength across the different coordinates to compensate for some of the missingness.  To see this, note that the MissCUSUM transformation is piecewise constant in each coordinate.  In particular, even in the noiseless setting, the absolute MissCUSUM series will typically not have a unique maximiser in each coordinate, but combining the information across coordinates allows us to pin down the changepoint location to an interval of length 
\[
\min\biggl\{t> z:\sum_{j=1}^p \omega_{j,t}\neq 0\biggr\} - \max\biggl\{t \leq z:\sum_{j=1}^p \omega_{j,t}\neq 0\biggr\}.
\]
This will often be a shorter interval than one would obtain from any of the individual component series.

The next step of the \texttt{MissInspect} algorithm is to use the MissCUSUM transformation to find a good projection direction $\hat{v} \in \mathbb{S}^{p-1}$.  The idea is that even though it is not possible to project the data along $\hat{v}$, due to the missingness, we can nevertheless compute the univariate series $\bigl((\hat{v}^\top T_\Omega)_t\bigr)_{t \in [n-1]}$, where $T_\Omega := \mathcal{T}^{\mathrm{Miss}}(X \circ \Omega,\Omega)$.  Writing $A_\Omega := \mathcal{T}^{\mathrm{Miss}}(\boldsymbol{\mu}\circ\Omega, \Omega)$ and $A:=\mathcal{T}(\boldsymbol{\mu})$, if each column of $X$ has identity covariance matrix, then for a generic projection direction $v \in \mathbb{S}^{p-1}$, we find that $\mathbb{E}\bigl\{(v^\top T_\Omega)_t \bigm| \Omega\bigr\} = (v^\top A_\Omega)_t$, and $\mathrm{Var} \bigl\{(v^\top T_\Omega)_t \bigm| \Omega\bigr\} = 1$.   In Proposition~\ref{Prop:Delta} below, we will show that $A_\Omega$ can be well approximated by the rank one matrix $(\diag \sqrt{\boldsymbol{q}})A = (\theta\circ \sqrt{\boldsymbol{q}}) \gamma^\top$, where 
\begin{align*}
\gamma &:= \frac{1}{\sqrt{n}}\biggl(\sqrt\frac{1}{n-1}(n-z),  \ldots,\sqrt\frac{z-1}{n-z+1}(n-z), \\
&\hspace{6cm}\sqrt{z(n-z)}, \sqrt\frac{n-z-1}{z+1}z,\ldots,\sqrt\frac{1}{n-1}z\biggr)^\top \in \mathbb{R}^{n-1}
\end{align*}
attains its peak in absolute value at the true changepoint location $z$.  Substituting this rank one approximation into the expression for $\mathbb{E}\bigl\{(v^\top T_\Omega)_t \bigm| \Omega\bigr\} = (v^\top A_\Omega)_t$ suggests that an oracle projection direction is a unit vector in the direction of $\theta\circ \sqrt{\boldsymbol{q}}$, which is the leading left singular vector of $(\diag \sqrt{\boldsymbol{q}})A$.  

For the corresponding problem with fully observed data, \citet{wang2016highdimensional} proposed a semi-definite relaxation technique to estimating the oracle projection direction.  Unfortunately, since $A_\Omega$ is not a rank one matrix when some data are missing, this relaxation turns out to be too coarse, and a new approach is required.  Motivated by the fact that $\theta\circ \sqrt{\boldsymbol{q}}$ has the same sparsity pattern as $\theta$, and viewing $T_\Omega$ as a perturbation of $(\diag \sqrt{\boldsymbol{q}}) A$, we propose to estimate the oracle projection direction by solving the following optimisation problem:
\begin{equation}
\label{Eq:Optimisation}
    (\hat v, \hat w) \in \argmax_{(\tilde{v},\tilde w) \in
\mathbb{B}^{p}\times \mathbb{B}^{n-1}} \bigl\{ \langle T_\Omega, \tilde{v}\tilde w^\top \rangle -   
\lambda\|\tilde{v}\|_1 \bigr\},
\end{equation}
where $\lambda > 0$ is a tuning parameter to be specified later. Here, with a suitable choice of $\lambda$, the $\ell_1$ penalty on $\tilde v$ in~\eqref{Eq:Optimisation} exploits the sparsity of the oracle projection direction to allow for consistent estimation of $(\theta\circ\sqrt{\boldsymbol{q}}) /
\|\theta\circ\sqrt{\boldsymbol{q}}\|_2$, even when the dimension $p$ is large, as will be shown in Proposition~\ref{Prop:SineAngle} in Section~\ref{Sec:Theory}.  A further advantage of~\eqref{Eq:Optimisation} over the semi-definite relaxation approach is that it directly exploits the row sparsity pattern of the rank one matrix $(\theta\circ\sqrt{\boldsymbol{q}})\gamma^\top$, as opposed to just the overall entrywise sparsity of this matrix.  Using the estimated oracle projection direction $\hat v$, we can project the MissCUSUM transformation $T_\Omega$ of $(X \circ \Omega,\Omega)$, and estimate the changepoint by the location of the maximum absolute value in the univariate projected series.  Pseudocode for the \texttt{MissInspect} algorithm is given in Algorithm~\ref{Algo:MissInspect}.

\begin{algorithm}[htbp]
\SetAlgoLined
\KwIn{$X_\Omega = X\circ\Omega \in \mathbb{R}^{p\times n}$, $\Omega \in
\{0,1\}^{p\times n}$, $\lambda > 0$}
  $T_\Omega \leftarrow \mathcal{T}^{\mathrm{Miss}}(X_\Omega,\Omega)$\; \label{Algo:MissInspectStep1}
 Find $(\hat{v},\hat{w}) \in \argmax_{\tilde{v} \in \mathbb{B}^{p-1},
\tilde w \in \mathbb{B}^{n-2}} \bigl\{ \langle T_\Omega, \tilde{v}\tilde
w^\top \rangle - \lambda\|\tilde{v}\|_1 \bigr\}$\;
\label{Algo:MissInspectStep2}
    $\hat{z} \leftarrow \mathrm{median} \bigl(\argmax_{t \in [n-1]} \bigl|(\hat{v}^\top T_\Omega)_t\bigr|\bigr)$\;
\KwOut{$\hat{z}$}
 \caption{\label{Algo:MissInspect}Pseudocode of the \texttt{MissInspect}
algorithm}
\end{algorithm}

The optimisation problem in Step~\ref{Algo:MissInspectStep2} of Algorithm~\ref{Algo:MissInspect} is bi-convex in $(\tilde v,\tilde w)$; i.e., the objective is concave in $\tilde v$ for every fixed $\tilde w$ and concave in $\tilde w$ for every fixed $\tilde v$. Hence, we can alternate between optimising over $\tilde v$ and $\tilde w$ in~\eqref{Eq:Optimisation}. By inspecting the Karush--Kuhn--Tucker
conditions as in Lemma~\ref{Lem:1}, we see that when $\lambda < \|T_\Omega\|_{2 \rightarrow \infty}$, both steps of each iteration have closed form expressions, which lead us to the iterative procedure to optimise~\eqref{Eq:Optimisation} given in Algorithm~\ref{Algo:PowerThresh}.  In that algorithm, we define the soft-thresholding function $\mathrm{soft}: \mathbb{R}^{p}\times [0,\infty) \rightarrow \mathbb{R}^p$ such that for $v = (v_1,\ldots,v_p)^\top \in\mathbb{R}^p$, we have $\bigl(\mathrm{soft}(v, \lambda)\bigr)_j = \mathrm{sgn}(v_j)\max\{|v_j| - \lambda, 0\}$ for $j\in [p]$.  We remark that $T_\Omega$ is known to the practitioner, so we can always choose $\lambda < \|T_\Omega\|_{2 \rightarrow \infty}$.  As usual for such iterative algorithms for bi-convex optimisation, the objective increases at each iteration; empirically, we have not observed any convergence issues.      

\begin{algorithm}[htbp]
\SetAlgoLined
\KwIn{$T_\Omega \in \mathbb{R}^{p \times (n-1)}$, $\lambda \in \bigl(0,\|T_\Omega\|_{2 \rightarrow \infty}\bigr)$}
  $\tilde{v} \leftarrow$ leading left singular vector of $T_\Omega$\;
 \Repeat{convergence}{
   $\tilde w \leftarrow
\frac{T_\Omega^\top\tilde{v}}{\|T_\Omega^\top\tilde{v}\|_2}$\;
   $\tilde v \leftarrow \frac{\mathrm{soft}(T_\Omega\tilde w,
\lambda)}{\|\mathrm{soft}(T_\Omega\tilde w, \lambda)\|_2}$\;}
\KwOut{$(\hat{v},\hat{w}) = (\tilde{v},\tilde{w})$}
 \caption{\label{Algo:PowerThresh}Pseudocode for an iterative procedure
optimising~\eqref{Eq:Optimisation}}
\end{algorithm}  

We conclude this section by illustrating the \texttt{MissInspect} algorithm in action in Figure~\ref{Fig:Illustration}.  Here, with $n=250$ and $p=100$, we generated $n$ independent $p$-variate Gaussian observations with mean structure~\eqref{Eq:SingleChangepoint} and identity covariance matrix.  We took $z=100$, and $\theta = (\vartheta \mathbf{1}_k/k^{1/2},\mathbf{0}_{p-k})^\top$, with $k=10$ and $\vartheta = 2$.  Thus, the first 10 coordinates represent signals, while the remaining 90 are noise coordinates.  All entries of our data matrix were observed independently (and independently of the data), with probability $0.2$.  The top panels display visualisations of the data and the MissCUSUM transformation respectively.  In the bottom-left panel, the coloured lines are the first five components of the MissCUSUM transformation; we see that these traces are piecewise constant, with jumps at observed data points.  Even though each of these five is obtained from a signal coordinate, the locations of the peaks of these individual series would not yield very reliable changepoint estimates, both because the noise introduces considerable variability (e.g.~the peak of the purple series runs from time 217 to 228), and because the missingness can lead to fairly long stretches where these series are constant.  Nevertheless, once all 100 series are aggregated appropriately by our \texttt{MissInspect} algorithm, the resulting black trace does have a sharper peak close to the true changepoint.  The bottom-right plot shows two nonparametric density estimates of the estimated changepoint locations from the \texttt{MissInspect} procedure over 1000 repetitions from this data generating mechanism; the first is a histogram, which requires the choice of a binwidth, while the second is the log-concave maximum likelihood estimator \citep{dumbgen2009maximum,cule2010maximum}, which is fully automatic.  Both indicate a sharp peak for the density close to the true changepoint; in the latter case, the mode is exactly at 100.   
\begin{figure}
    \centering
    \includegraphics[width=0.9\textwidth]{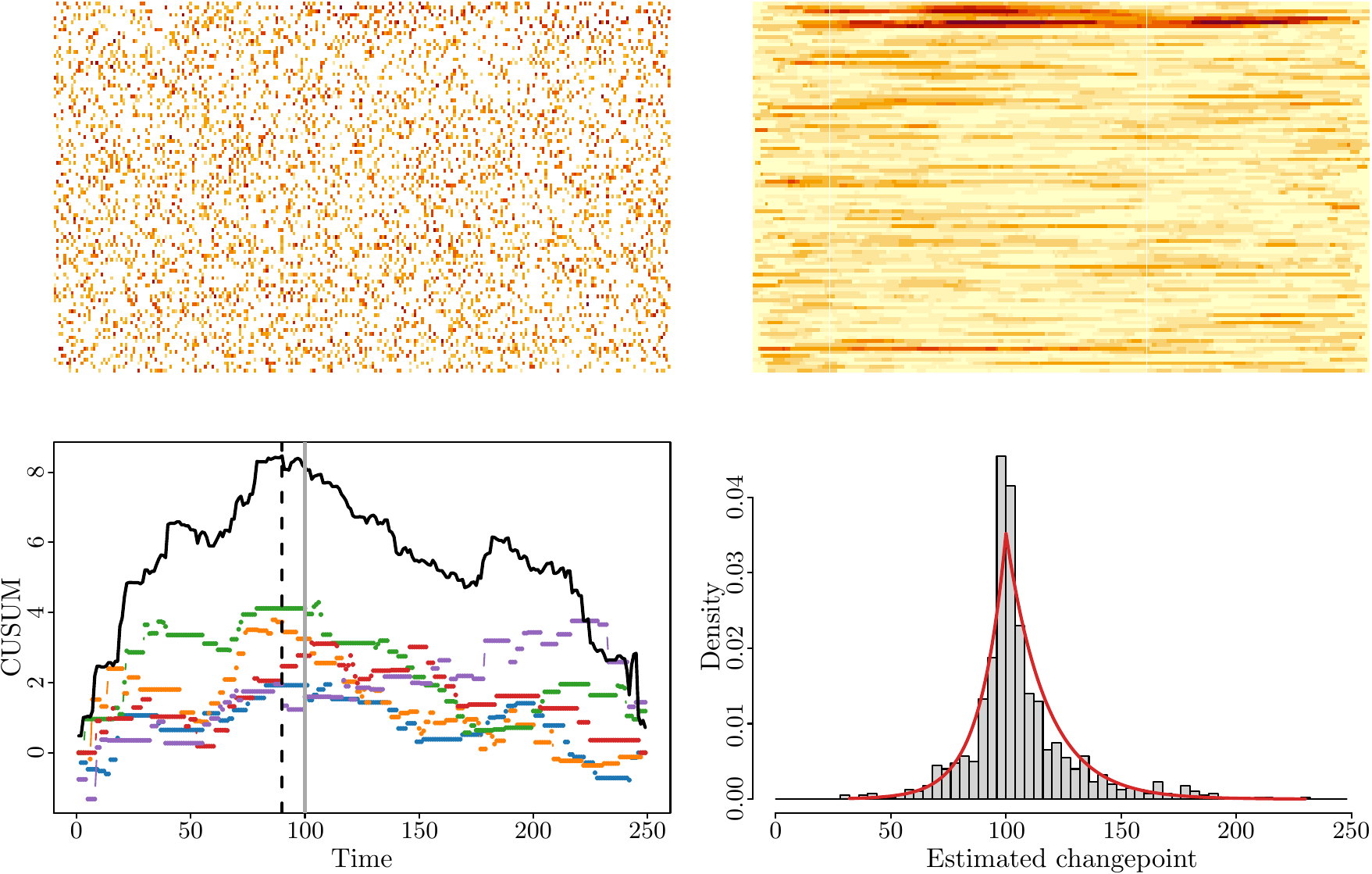}
    \caption{\small\texttt{MissInspect} algorithm in action. Top-left: visualisation of the data matrix with $p=100$ and $n=250$, where each column represents a $p$-dimensional observation and missing entries are shown in white.  Darker colours indicate larger values.  Time runs from left to right, and a change in mean occurs at time $100$ in each of the first ten rows. Top-right: visualisation of the MissCUSUM transformation of the data.  Bottom-left: the first five rows of the MissCUSUM matrix are plotted in colour, and the black curve shows the projected MissCUSUM series, which is maximised at the estimated changepoint location of $90$ (black dashed line). The true changepoint is shown as a grey solid line. Bottom-right: histogram of estimated changepoints over 1000 repetitions from the same data setting; a log-concave estimated density is shown in red.}
    \label{Fig:Illustration}
\end{figure}

\section{Theoretical guarantees}
\label{Sec:Theory}
We will focus our theoretical analysis on the single changepoint setting, in order to try to articulate more clearly the way that the coordinate-wise signal-to-noise ratio and missingness mechanism interact to determine both the performance of the \texttt{MissInspect} algorithm and the fundamental difficulty of the problem.  Moreover, we assume that the revelation matrix $\Omega = (\omega_{j,t}) \in \{0,1\}^{n \times p}$ has a \emph{row-homogeneous} distribution, in the sense that there exists a vector $\boldsymbol{q} = (q_1,\ldots,q_p)^\top \in (0,1]^p$ such that $\omega_{j,t}\sim \mathrm{Bern}(q_j)$, independently for all $j \in [p]$ and $t \in [n]$. We will refer to $\boldsymbol{q}$ as the \emph{observation rate vector}. Such a row-homogeneous assumption may be appropriate, for instance, in applications where each component series is measured by a separate device with its own observation rate.  As for the data, we will assume that the columns $(X_t)_{t \in [n]}$ of the data matrix $X = (X_{j,t})_{j \in [p], t \in [n]} \in \mathbb{R}^{n \times p}$ satisfy
\begin{equation}
\label{Eq:Xt}
X_t \sim \mathcal{N}_p(\mu_t,\sigma^2 I_p), \ \text{independently for $t \in [n]$},
\end{equation}
where $\mu_1,\ldots,\mu_n$ satisfy~\eqref{Eq:SingleChangepoint}.   

For $n \in \mathbb{N}$, $z\in[n-1]$, $\theta = (\theta_1,\ldots,\theta_p)^\top \in\mathbb{R}^p$, $\sigma>0$ and $\boldsymbol{q} = (q_1,\ldots,q_p)^\top \in(0,1]^p$ we write $P_{n,p,z,\theta,\sigma,\boldsymbol{q}}$ for the joint distribution of $(X, \Omega)$, where $X$ and $\Omega$ are independent, where $X$ satisfies~\eqref{Eq:Xt} with the vector of mean change $\theta := \mu^{(2)} - \mu^{(1)} \in \mathbb{R}^p$ satisfying $\|\theta\|_0 \leq k$, and where $\Omega$ has a row-homogeneous distribution with observation rate vector $\boldsymbol{q} = (q_1,\ldots,q_p)^\top \in (0,1]^p$.  We write $\tau := n^{-1}\min(z,n-z)$.  Recall our notation  $\|\theta\|_{2,\boldsymbol{q}}^2 := \sum_{j=1}^p \theta_j^2q_j$, a quantity that captures a key interaction between the signal strength and observation rate.   Our first result below shows that the projection direction $\hat v$ obtained from Step~\ref{Algo:MissInspectStep2} of Algorithm~\ref{Algo:MissInspect} is closely aligned with~$\theta\circ \sqrt{\boldsymbol{q}}$, which, as argued in Section~\ref{Sec:Methodology}, can be regarded as an oracle projection direction. 

\begin{proposition}
\label{Prop:SineAngle}
Let $(X,\Omega)\sim P_{n,p,z,\theta,\sigma,\boldsymbol{q}}$ and let $(\hat{v}, \hat{w})$ be obtained from Step~\ref{Algo:MissInspectStep2} in Algorithm~\ref{Algo:MissInspect}, applied with inputs $X_\Omega = X \circ \Omega$, $\Omega$ and $\lambda \geq 2 \sigma \sqrt{n\log(pn)}$.  Then 
\[
 \mathbb{P}\biggl\{ \sin \angle (\hat{v}, \theta \circ \sqrt{\boldsymbol{q}} ) \leq \frac{32\lambda \sqrt{k}}{n\tau\|\theta\|_{2,\boldsymbol{q}} } + \frac{112\|\theta\|_{2}}{\tau \|\theta\|_{2,\boldsymbol{q}}}\sqrt\frac{6\log(kn)}{n} \biggr\} \geq 1- \frac{6}{kn}.
\]
\end{proposition}
Considering the case $\lambda = 2 \sigma \sqrt{n\log(pn)}$ for simplicity, Proposition~\ref{Prop:SineAngle} reveals that, with high probability, the sine of the acute angle between $\hat{v}$ and $\theta\circ \sqrt{\boldsymbol{q}}$ is controlled by the sum of two terms: the first of these represents the estimation error caused by the noise in the data we observe, and we see that $\|\theta\|_{2,\boldsymbol{q}}/\sigma$ can be thought of as an effective signal-to-noise ratio.  On the other hand, the second term reflects the error due to our incomplete observations (and would be present even in the noiseless case with $\sigma = 0$); here $\|\theta\|_{2,\boldsymbol{q}}^2/\|\theta\|_2^2$ may be regarded as a \emph{signal-weighted observation probability}.

From a theoretical point of view, the fact that $\hat{v}$ is estimated using the entire available data set $X_\Omega$ makes it difficult to analyse the post-projection noise structure. For this reason, in the analysis below, we work with a sample-splitting variant of Algorithm~\ref{Algo:MissInspect}, as given in Algorithm~\ref{Algo:MissInspectVariant}.  Here, the projection direction $\hat{v}$ is estimated using only the observed data at odd-numbered time points, and the MissCUSUM transformation of the observed data at even-numbered time points is then projected along $\hat{v}$ to obtain the final estimate of the changepoint location.       

\begin{algorithm}[htbp]
\SetAlgoLined
\KwIn{$X_\Omega = X\circ\Omega \in \mathbb{R}^{p \times n}$, $\Omega \in \{0,1\}^{p \times n}$, $\lambda > 0$} 
$n_1 \leftarrow \lfloor n/2\rfloor$\;
Let $\Omega^{(1)} \in \{0,1\}^{p \times n_1}$ and $\Omega^{(2)} \in \{0,1\}^{p \times n_1}$ denote the matrices formed from the first $n_1$ odd and the $n_1$ even numbered columns of $\Omega$ respectively\;
Let $X_{\Omega}^{(1)} \in \mathbb{R}^{p \times n_1}$ and $X_{\Omega}^{(2)} \in \mathbb{R}^{p \times n_1}$ denote the matrices formed from the first $n_1$ odd and the $n_1$ even numbered columns of $X_\Omega$ respectively\;
  $T_{\Omega}^{(1)} \leftarrow \mathcal{T}^{\mathrm{Miss}}(X_{\Omega}^{(1)},\Omega^{(1)}) \in \mathbb{R}^{p \times (n_1-1)}$\; \label{Step:5}
    $T_{\Omega}^{(2)} \leftarrow \mathcal{T}^{\mathrm{Miss}}(X_{\Omega}^{(2)},\Omega^{(2)}) \in \mathbb{R}^{p \times (n_1-1)}$\;
    $(\hat{v},\hat{w}) \leftarrow \argmax_{\tilde{v} \in \mathbb{B}^{p-1},\tilde{w} \in \mathbb{B}^{n-2}} \bigl\{\langle T_{\Omega}^{(1)},\tilde{v}\tilde{w}^\top \rangle - \lambda\|\tilde{v}\|_1\bigr\}$\; \label{Step:7}
    $\hat{z} \leftarrow 2\, \mathrm{median}\bigl( \argmax_{t \in [n_1-1]} \bigl|(\hat{v}^\top T_{\Omega}^{(2)})_t\bigr|\bigr)$\;
\KwOut{$\hat{z}$}
 \caption{\label{Algo:MissInspectVariant}Pseudo-code for the sample-splitting variant of the \texttt{MissInspect} algorithm.}
\end{algorithm}

Theorem~\ref{Thm:SlowRate} is our first main result on the performance of Algorithm~\ref{Algo:MissInspectVariant} in contexts where our data are generated from a single changepoint, row-homogeneous model $P_{n,p,z,\theta,\sigma,\boldsymbol{q}}$. 
\begin{theorem}
\label{Thm:SlowRate}
Suppose $(X,\Omega)\sim P_{n,p,z,\theta,\sigma,\boldsymbol{q}}$.  Assume for simplicity that $n$ and $z$ are even.  Let $\hat{z}$ be the output of Algorithm~\ref{Algo:MissInspectVariant} with inputs $X\circ \Omega, \Omega$ and $\lambda = 2\sigma\sqrt{n\log(pn)}$.  There exist universal constants $C,C' > 0$ such that whenever
  \begin{equation}
 \label{Eq:C'}
  \frac{C'}{\tau}\sqrt{\frac{\log(pn)}{n}}\biggl(\frac{\sigma \sqrt{k}}{\|\theta\|_{2,\boldsymbol{q}}} + \frac{\|\theta\|_2}{\|\theta\|_{2,\boldsymbol{q}}}\biggr) \leq \frac{1}{2},
\end{equation}
we have
\[
    \mathbb{P}\biggl\{\frac{|\hat{z}-z|}{n\tau} > C\sqrt{\frac{\log(kn)}{n\tau}}\biggl(\frac{\sigma}{\|\theta\|_{2,\boldsymbol{q}}} + \frac{\|\theta\|_2}{\|\theta\|_{2,\boldsymbol{q}}}\biggr)\biggr\} \leq \frac{22}{n}.
  \]
\end{theorem}
Condition~\eqref{Eq:C'} ensures that the projection direction $\hat{v}$ obtained in Step~\ref{Step:7} of Algorithm~\ref{Algo:MissInspectVariant} has non-trivial correlation with the oracle projection direction $\theta \circ \sqrt{\boldsymbol{q}}$; cf.~Proposition~\ref{Prop:SineAngle}.  An attractive feature of Theorem~\ref{Thm:SlowRate} is the way that the interaction between the signal strength and the observation rate is captured through $\|\theta\|_{2,\boldsymbol{q}}$.  As mentioned in the introduction, this weighted average provides much greater understanding of the influence of missingness on the performance of the \texttt{MissInspect} algorithm than more naive bounds that depend on the worst-case missingness probability across all rows.  For instance, we see that a high degree of missingness in noise or weak signal coordinates may not have too much of a detrimental effect on performance compared with complete observation of these coordinates.  See also Theorem~\ref{Thm:LowerBound} below for confirmation of the way in which $\|\theta\|_{2,\boldsymbol{q}}$ also controls the fundamental difficulty of the problem (not just for our procedure).  

A further attraction of Theorem~\ref{Thm:SlowRate} is the absence of any condition on the number of observations in each row.  On the other hand, it turns out that if the expected number of observations in each row is at least $k/\tau^2$ (up to logarithmic factors), then we can obtain a substantially improved bound on the rate of estimation of Algorithm~\ref{Algo:MissInspectVariant}.

\begin{theorem}
\label{Thm:FastRate}
Suppose $(X,\Omega)\sim P_{n,p,z,\theta,\sigma,\boldsymbol{q}}$.  Assume for simplicity that $n$ and $z$ are even.  Let $\hat{z}$ be the output of Algorithm~\ref{Algo:MissInspectVariant} with inputs $X\circ \Omega,\Omega$ and $\lambda = 2\sigma\sqrt{n\log(pn)}$. Define
\[
\rho:=  \frac{1}{\tau}\sqrt{\frac{\log(pn)}{n}}\biggl(\frac{\sigma \sqrt{k}}{\|\theta\|_{2,\boldsymbol{q}}} + \frac{\|\theta\|_2}{\|\theta\|_{2,\boldsymbol{q}}}\biggr).
\]
Then there exist universal constants $c, C_1,C_2 > 0$ such that if $\rho \leq c$ and $n\tau^2 \min_{j\in[p]} q_j \geq C_1k\log(pn)$, then
\[
\mathbb{P}\biggl\{\frac{|\hat z - z|}{n\tau} >  \frac{C_2\log(pn)}{n\tau}\biggl(\frac{\sigma^2}{\|\theta\|_{2,\boldsymbol{q}}^2}+ \frac{\|\theta\|_{\infty}^2}{\|\theta\|_{2,\boldsymbol{q}}^2}\biggr)\biggr\} \leq \frac{23}{n}. 
\]
\end{theorem}
The rate obtained in Theorem~\ref{Thm:FastRate} is essentially the square of that obtained in Theorem~\ref{Thm:SlowRate}.  In fact, an additional improvement is the reduction of $\|\theta\|_2$ in the second term to $\|\theta\|_\infty$.  Again, we see the decomposition of the estimation error into terms reflecting the noise in the observed data and the incompleteness of the observations respectively.

As a complement to Theorem~\ref{Thm:FastRate}, we now present a minimax lower bound, which studies the fundamental limits of the expected estimation error that are achievable by any algorithm.  We write $\tilde{\mathcal{Z}}$ for the set of estimators of $z$, i.e.~the set of Borel measurable functions $\hat{z}:\mathbb{R}^{n \times p} \times \{0,1\}^{n \times p} \rightarrow [n-1]$.

\begin{theorem}
\label{Thm:LowerBound}
Let $M \geq 1$ satisfy $\|\theta\|_\infty\leq M\min_{j\in [p]:\theta_j\neq 0}|\theta_j|$. If $\max\{\sigma^2,\|\theta\|_\infty^2/(2M^2)\}\geq \|\theta\|_{2,\boldsymbol{q}}^2$, then there exists $c>0$, depending only on $M$, such that for $n \geq 4$,
\[
\inf_{\tilde z \in \hat{\mathcal{Z}}}\max_{z \in [n-1]}\mathbb{E}_{P_{n,p,z,\theta,\sigma,\boldsymbol{q}}} \frac{|\tilde z(X\circ\Omega, \Omega) - z|}{n\tau} \geq \frac{c}{n\tau}\min\biggl\{\frac{\sigma^2}{\|\theta\|_{2,\boldsymbol{q}}^2} + \frac{\|\theta\|_\infty^2}{\|\theta\|_{2,\boldsymbol{q}}^2},n\biggr\}.
\]
\end{theorem}
Theorem~\ref{Thm:LowerBound} reveals that the \texttt{MissInspect} algorithm as given in Algorithm~\ref{Algo:MissInspectVariant} attains the minimax optimal estimation error rate up to logarithmic factors in all of the parameters of the problem, at least in settings where the signals are of comparable magnitude.  Note that the \texttt{MissInspect} algorithm also matches (deterministically) the second term in the minimum in Theorem~\ref{Thm:LowerBound}, because it trivially satisfies $|\hat{z} - z| \leq n-2$.  The form of the lower bound in Theorem~\ref{Thm:LowerBound} confirms that $\|\theta\|_{2,\boldsymbol{q}}$ is the correct functional of the mean change vector~$\theta$ and observation $\boldsymbol{q}$ for capturing the difficulty of the changepoint estimation problem in our missing data setting.

\section{Numerical studies}
\label{Sec:Numerical}
\subsection{Choice of tuning parameter}
The tuning parameter choice of $\lambda = 2\sigma\sqrt{n\log(pn)}$ is convenient in our theoretical analysis. However, this choice often turns out to be slightly too conservative in practice, so to explore this, we considered the output of Algorithm~\ref{Algo:PowerThresh} for a range of $\lambda$ values, under several different settings of $n$, $p$, $k$, $\theta$ and $\boldsymbol{q}$.    Figure~\ref{Fig:Lambda} displays the mean angle between the estimated projection direction $\hat v$ from~\eqref{Eq:Optimisation} and the oracle projection direction $\theta \circ \sqrt{\boldsymbol{q}} / \|\theta \circ \sqrt{\boldsymbol{q}}\|_2$ as a function of $\lambda$ in two such sets of simulations. In both panels, we set $n=1000$, $p=500$, $z=400$ and took $\lambda = a\sigma\sqrt{n\log(pn)}$ for $a\in [0,2]$.  The vector of mean change is $\theta = \vartheta k^{-1/2}(\mathbf{1}_k^\top, \mathbf{0}_{p-k}^\top)^\top$. Data were observed according to the row-homogeneous missingness model with $q_j = q^{\mathrm{s}}$ if $\theta_j\neq 0$ and $q_j = q^{\mathrm{n}}$ otherwise. In the left panel of the figure, we set $q^{\mathrm{s}} = q^{\mathrm{n}} = 0.2$ and vary $k\in\{3,10,50\}$ and $\vartheta\in\{1,1.5,2,2.5,3\}$, whereas in the right panel, we took $k=3$, $\vartheta=2$, $\sigma=1$ and vary $q^{\mathrm{s}}, q^{\mathrm{n}} \in \{0.1,0.2,0.3,0.4,0.5\}$.  We note that the choice $\lambda = 2^{-1}\sigma\sqrt{n\log(pn)}$ performs well in all settings, especially when the signal is relatively sparse.  We therefore settled on  this choice of $\lambda$ throughout our numerical studies.  It is reassuring to see from the right panel of Figure~\ref{Fig:Lambda} that the performance of the projection direction estimator has almost no dependence on $q^{\mathrm{n}}$, as predicted by our Proposition~\ref{Prop:SineAngle} (since $\|\theta\|_{2,\boldsymbol{q}}$ does not depend on~$q^{\mathrm{n}}$).

\begin{figure}[htbp]
\centering
\includegraphics[width=0.9\textwidth]{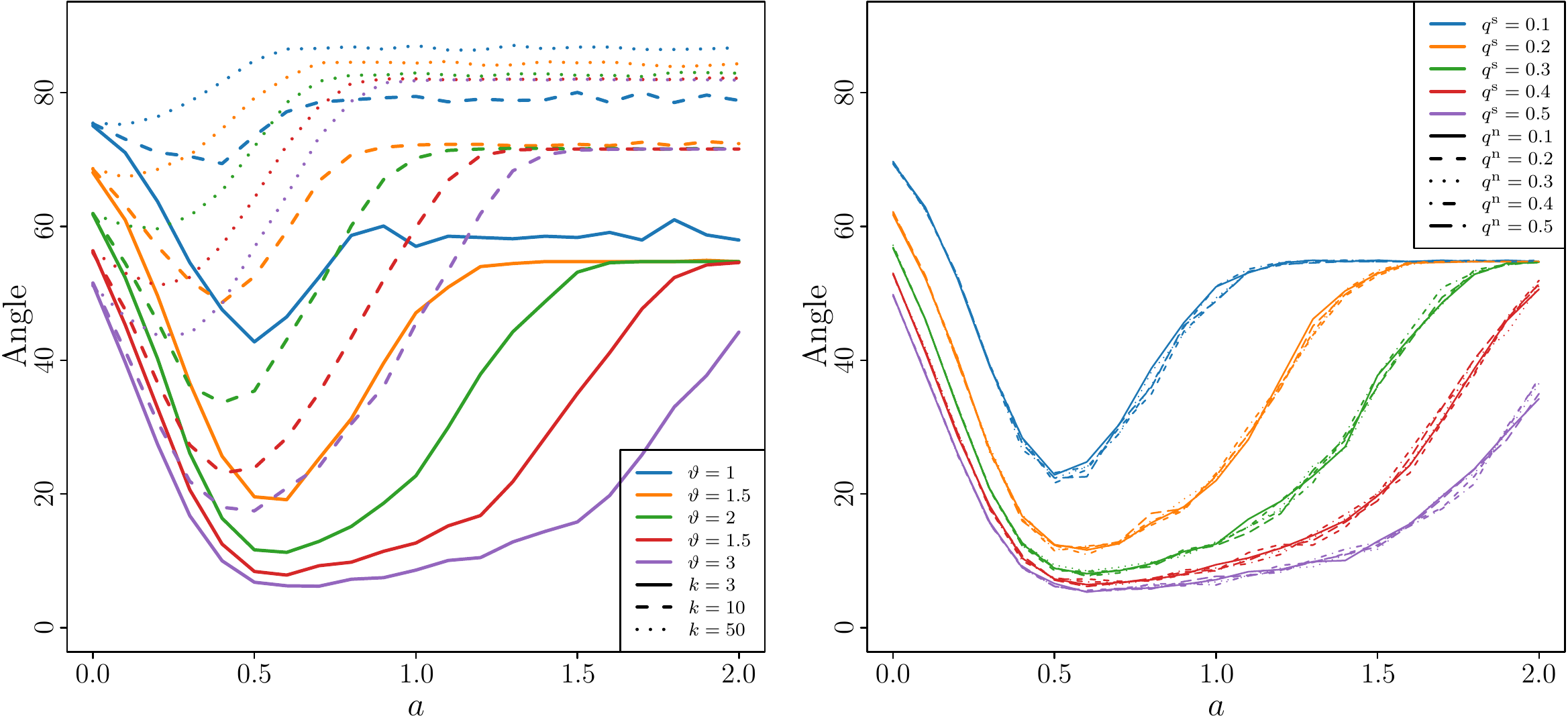}
\caption{\small Mean angle in degrees (averaged over 200 repetitions) between the oracle projection direction and the estimated projection direction from~Algorithm~\ref{Algo:PowerThresh} with $\lambda=a\sigma\sqrt{n\log(pn)}$ for $a\in [0, 2]$. Data are generated under~\eqref{Eq:Xt} with row-homogeneous missingness, independent of the data. Parameters: $n=1000$, $p=500$, $z=400$, $\sigma=1$. Left panel: all entries observed independently with probability $0.2$, $k\in\{3,10,50\}$ and $\vartheta\in\{1,1.5,2,2.5,3\}$. Right panel: $k=3$, $\vartheta=2$, signal coordinates are observed with probability $q^{\mathrm{s}} \in\{0.1,0.2,0.3,0.4,0.5\}$ and noise coordinates are observed with probability $q^{\mathrm{n}} \in\{0.1,0.2,0.3,0.4,0.5\}$.}
\label{Fig:Lambda}
\end{figure}

\subsection{Validation of theoretical results}

The aim of this subsection is to provide empirical confirmation of the forms of the bounds obtained in Proposition~\ref{Prop:SineAngle} and Theorem~\ref{Thm:FastRate}.  In particular, we would like to verify that the crucial quantity $\|\theta\|_{2,\boldsymbol{q}}$ does indeed capture the appropriate interaction between signal and missingness that determines the performance of the \texttt{MissInspect} algorithm.  The two panels of Figure~\ref{Fig:Validation} study the angle between the estimated and oracle projection directions, and the estimated changepoint location error respectively.  To obtain this figure, we set $n=1200$, $p=1000$ and generated data vectors under~\eqref{Eq:Xt} with every entry observed independently with probability $q \in \{0.1,0.2,0.4,0.8\}$, independent of the data.  A single change occurred at $z=400$ with vector of mean change $\theta = \vartheta k^{-1/2}(\mathbf{1}_k^\top, \mathbf{0}_{p-k}^\top)^\top$ and $k=3$. We investigated the performance of \texttt{MissInspect} over 200 Monte Carlo repetitions for each of $\vartheta \in \{0.5,1,1.5,2\}$ and $\sigma \in \{0.2,0.4,0.8,1.6\}$. The left panel of Figure~\ref{Fig:Validation} shows that the logarithm of the mean sine angle loss decreases approximately linearly with $\log \|\theta\|_{2,\boldsymbol{q}}$, with gradient approximately $-1$.
This is consistent with the conclusion of Proposition~\ref{Prop:SineAngle}, which shows that the sine angle loss is controlled with high probability by an upper bound that is inversely proportional to $\|\theta\|_{2,\boldsymbol{q}} = \vartheta q^{1/2}$.  Moreover, curves corresponding to $\sigma=0.2, 0.4, 0.8, 1.6$ are roughly equally spaced on the logarithmic scale, which corresponds to the linear dependence on $\lambda = 2^{-1}\sigma\sqrt{n\log(pn)}$ of the first term in the high-probability bound in Proposition~\ref{Prop:SineAngle}.  For fixed $\sigma$, the blue, orange and green curves are approximately overlapping, especially when the sine angle loss is large.  In particular, doubling $\vartheta$ and reducing $q$ by a factor of four leaves the sine angle loss virtually unchanged in these settings, which is consistent with the first term in the high probability upper bound in Proposition~\ref{Prop:SineAngle} being the dominant one, with its reciprocal dependence on $\|\theta\|_{2,\boldsymbol{q}} = \vartheta q^{1/2}$.  The contribution from the second term in Proposition~\ref{Prop:SineAngle} is still visible in the high signal-to-noise ratio settings, where, for instance when $\sigma = 0.2$, the $\vartheta=2$ curve (green) lies above the $\vartheta=1$ curve (orange).  This is again consistent with the form of the second term in the bound in Proposition~\ref{Prop:SineAngle}, which, in our setting, does not depend on $\vartheta$, but is inversely proportional to $q^{1/2}$.

\begin{figure}[htbp]
\centering
\includegraphics[width=0.9\textwidth]{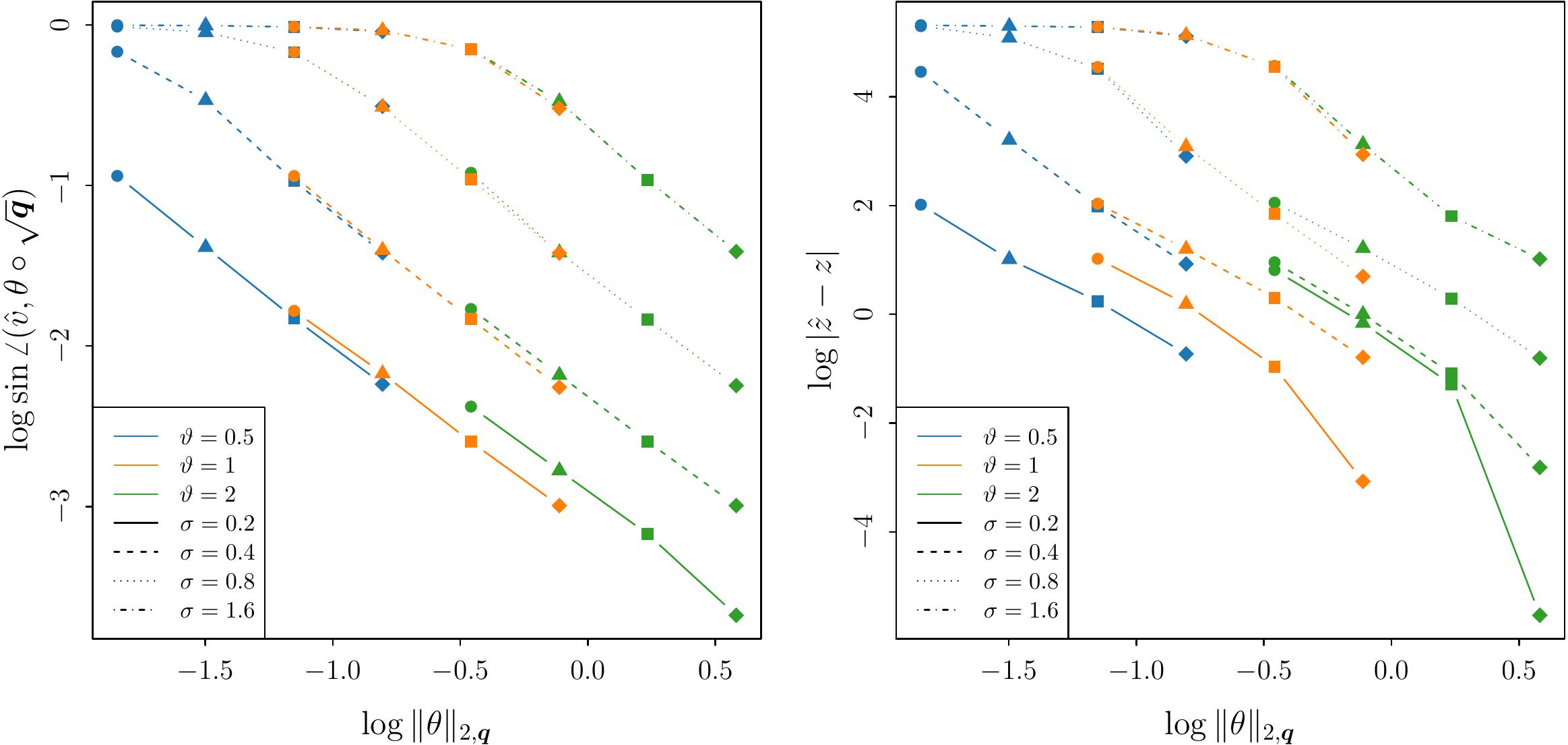}
\caption{\small Estimation accuracy of \texttt{MissInspect} as a function of  $\|\theta\|_{2,\boldsymbol{q}}$. Data are generated under~\eqref{Eq:Xt} with row-homogeneous missingness, independent of the data. Parameters: $n=1200$, $p=1000$, $z=400$, $k=3$, $\vartheta\in\{0.5,1,2\}$, $\sigma\in\{0.2,0.4,0.8,1.6\}$ and $\boldsymbol{q} = q\mathbf{1}_p$ with $q\in\{0.1,0.2,0.4,0.8\}$.  Colours indicate $\vartheta$, line type indicates $\sigma$, and the circle, triangle, square and diamond plotting characters correspond to $q = 0.1, 0.2, 0.4, 0.8$ respectively.  Left panel: logarithm of the mean sine angle loss (averaged over 200 repetitions) between estimated and oracle projection directions. Right panel: logarithm of mean changepoint location loss (averaged over 200 repetitions).}
\label{Fig:Validation}
\end{figure}

A similar story emerges in the right panel of Figure~\ref{Fig:Validation} for the changepoint location estimator accuracy. Here, for fixed $\vartheta$ and $\sigma$, most points lie on approximate straight lines with slope $-2$, which is in agreement with the $\|\theta\|_{2,\boldsymbol{q}}^{-2}$ dependence in the high probability bound of $|\hat z - z|$ in Theorem~\ref{Thm:FastRate}. The $\sigma^2$ dependence in the first term of the bound in Theorem~\ref{Thm:FastRate} is represented by the mostly equi-spaced curves for the four different equi-spaced $\sigma$ values on the logarithmic scale. The contribution of the second term in the bound in Theorem~\ref{Thm:FastRate} can be seen from the three curves corresponding to the smallest noise scale $\sigma = 0.2$. Here, the estimation error only improves slightly as $\vartheta$ increases, which is in agreement with our theoretical prediction, since, in the setting of this simulation, the second term in Theorem~\ref{Thm:FastRate} is proportional to $q^{-1/2}$ and does not depend on $\vartheta$.

\begin{table}[htbp]
    \centering
\begin{tabular}{ccccccc}
\hline\hline
$\nu$ & $k$ & $\vartheta$ & $\angle(\hat v^{\mathrm{MI}}, \theta \circ \sqrt{\boldsymbol{q}})$ & $\angle(\hat v^{\mathrm{II}}, \theta )$ & $|\hat z^{\mathrm{MI}} - z|$& $|\hat z^{\mathrm{II}} - z|$\\
\hline
$0.1$ & $3$ & $1$ & $70.0$ & $88.0$ & $143.9$ & $451.1$\\
$0.1$ & $3$ & $2$ & $41.9$ & $59.6$ & $44.8$ & $295.2$\\
$0.1$ & $3$ & $3$ & $26.2$ & $41.2$ & $13.0$ & $247.1$\\
$0.1$ & $44$ & $1$ & $83.4$ & $88.5$ & $183.9$ & $448.1$\\
$0.1$ & $44$ & $2$ & $64.4$ & $85.0$ & $77.5$ & $410.9$\\
$0.1$ & $44$ & $3$ & $48.7$ & $73.2$ & $15.1$ & $298.4$\\
$0.1$ & $2000$ & $1$ & $86.6$ & $88.2$ & $196.9$ & $445.7$\\
$0.1$ & $2000$ & $2$ & $77.3$ & $87.8$ & $124.8$ & $449$\\
$0.1$ & $2000$ & $3$ & $68.6$ & $82.9$ & $58.5$ & $384.6$\\
$0.5$ & $3$ & $1$ & $33.1$ & $81.9$ & $11.9$ & $362.2$\\
$0.5$ & $3$ & $2$ & $13.6$ & $40.8$ & $1.6$ & $7.4$\\
$0.5$ & $3$ & $3$ & $9.3$ & $22.8$ & $0.7$ & $4.0$\\
$0.5$ & $44$ & $1$ & $63.7$ & $88.3$ & $58.9$ & $440.7$\\
$0.5$ & $44$ & $2$ & $37.3$ & $72.0$ & $2.3$ & $129.1$\\
$0.5$ & $44$ & $3$ & $27.1$ & $58.0$ & $0.7$ & $1.5$\\
$0.5$ & $2000$ & $1$ & $77.4$ & $88.7$ & $112.6$ & $446.1$\\
$0.5$ & $2000$ & $2$ & $59.3$ & $85.8$ & $9.4$ & $357.4$\\
$0.5$ & $2000$ & $3$ & $52.1$ & $72.6$ & $1.6$ & $46.6$\\
\hline\hline
\end{tabular}
    \caption{\small Location and projection direction estimation errors (averaged over 200 Monte Carlo repetitions) for \texttt{MissInspect}  (denoted by superscript MI) and \texttt{ImputeInspect} (denoted by superscript II). Other parameters: $n=1200$, $p=2000$, $z=400$, $q_1,\ldots,q_p\stackrel{\mathrm{iid}}\sim \mathrm{Beta}\bigl(10\nu , 10(1-\nu)\bigr)$.}
    \label{Tab:Comparison}
\end{table}

\subsection{Comparison with a competitor}

Since we are not aware of other methods that have been proposed for the problem studied in this paper, in this subsection, we compare the performance of \texttt{MissInspect} with a natural alternative that combines the idea of handling missing data via imputation and the original \texttt{inspect} procedure.   Specifically, given a data matrix with incomplete observations, this comparator first applies the \texttt{softImpute} procedure of \citet{mazumder2010spectral}, with the maximum matrix rank parameter set to $2$ (since $X\circ \Omega$ can be viewed as a perturbation of its mean $(\diag \sqrt{\boldsymbol{q}}) \boldsymbol{\mu}$, which has rank 2). It then performs changepoint estimation on the imputed data matrix using the \texttt{inspect} procedure of \citet{wang2016highdimensional}, with the suggested regularisation parameter choice therein. We refer to this alternative approach as \texttt{ImputeInspect}.  Table~\ref{Tab:Comparison} compares the performance of \texttt{MissInspect} and \texttt{ImputeInspect} under various settings. Here, we choose $n=1200$, $p=2000$, $k\in\{3,\lfloor \sqrt p\rfloor, p\}$, $\vartheta\in\{1,2,3\}$. Data are generated according to~\eqref{Eq:Xt} with row-homogeneous missingness, independent of the data.  The changepoint occurs at $z=400$, with vector of mean change having $\ell_2$ norm~$\vartheta$ and proportional to $(1, 2^{-1/2}, \ldots, k^{-1/2}, 0,\ldots,0)^\top$. The observation rate vector $\boldsymbol{q} = (q_1,\ldots,q_p)^\top$ is randomly generated, independent of all other sources of randomness, such that $q_j \stackrel{\mathrm{iid}}{\sim} \mathrm{Beta}\bigl(10\nu, 10(1-\nu)\bigr)$ for $j\in[p]$, where $\nu \in\{0.1,0.5\}$. Since both the \texttt{ImputeInspect} and \texttt{MissInspect} procedures are projection based, it is natural to compare their performance in both the projection direction estimation and the final changepoint estimation errors. Note that the oracle projection direction for \texttt{ImputeInspect} is parallel to~$\theta$ (since the imputed matrix has no missing entries), whereas the oracle projection direction of \texttt{MissInspect} is parallel to $\theta\circ \sqrt{\boldsymbol{q}}$. We see in Table~\ref{Tab:Comparison} that \texttt{MissInspect} consistently outperforms \texttt{ImputeInspect}, often dramatically, for all observation fractions, sparsity levels and signal strengths considered.

\subsection{Real data analysis}
\label{Sec:RealData}

In this subsection, we illustrate the applicability of the \texttt{MissInspect} algorithm on an oceanographic data set covering the Neogene geological period.  Oceanographers study historic changes in the global ocean circulation system by examining microfossils that record the isotopic composition of water at the time at which they lived \citep{wright1996control}.  In particular, large cores are extracted from the ocean floor and a species of microfossils called foraminifera are taken from small slices of sediment at different depths within the core.  The ratio of the abundances of \textsuperscript{13}C to \textsuperscript{12}C isotopes in their calcium carbonate shells is compared against a standard, to understand the carbon composition within the oceans during their lifetime, and hence to determine the direction in which ocean currents flowed.  The depth of the foraminifera within the core is used as a proxy for the geological age of the fossil, measured in millions of years (Ma).

Our data, which are available in a GitHub repository\footnote{\url{https://github.com/wangtengyao/MissInspect/real_data}}, and were previously analysed by \citet{samworth2005understanding} and \citet{poore2006neogene}, consist of measurements from 16 cores extracted from the North Atlantic, Pacific and Southern Oceans and are displayed in Figure~\ref{Fig:Oceans}.  In total, there are 7369 observations at 6295 distinct time points, but Figure~\ref{Fig:Oceans} makes clear that the heterogeneous nature of the data collection process means that it is appropriate to think of the data as containing missingness.  The figure also indicates the 10 most prominent changepoints identified by applying the \texttt{MissInspect} algorithm in combination with binary segmentation, as discussed in Section~\ref{Sec:Extensions} below.  It is notable that the first changepoint, found by applying the algorithm to the full data set, occurs at 6.13Ma, a time that has previously been identified as a time of rapid change in oceanographic current flow \citep[][p.~13]{poore2006neogene}.
\begin{figure}
    \centering
    \includegraphics[width=0.9\textwidth]{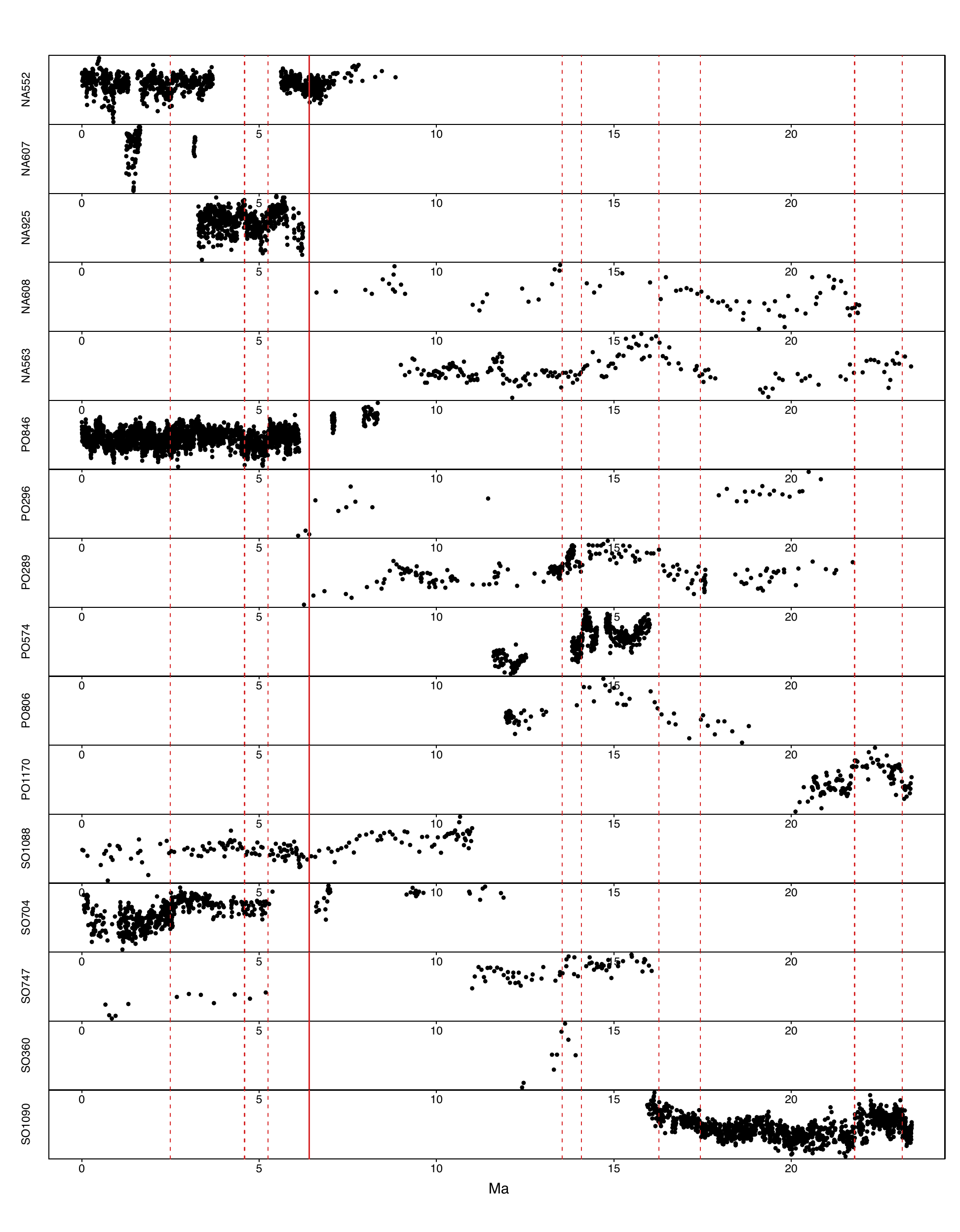}
    \caption{\small Ratios of carbon isotope measurements taken from foraminifera in 16 different cores from the North Atlantic, Pacific and Southern Oceans.  The label of each panel indicates both the ocean and the number of the core, while the horizontal axis measures geological time (0--23 Ma).  The red dashed lines indicate the 10 most prominent changepoints identified by applying the \texttt{MissInspect} algorithm in combination with binary segmentation, with the most significant change plotted with a solid line.}
    \label{Fig:Oceans}
\end{figure}
 
\section{Extensions}
\label{Sec:Extensions}

As mentioned in the introduction, one of our main theoretical goals in this work is to understand the way in which the missingness and the signal interact in changepoint problems to determine the difficulty of the problem.  This is particularly challenging when we seek to handle both high dimensionality and different levels of missingess in different coordinates.  For the purposes of our theoretical analysis, then, it is natural to impose stronger assumptions elsewhere, so as to best expose the interesting phenomena at play.  Nevertheless, it remains of interest to consider the extent to which the methodology could be generalised, and the assumptions could be relaxed, to cover a wider range of scenarios and problems one might see in practice.

As we saw in analysing the oceanography data in Section~\ref{Sec:RealData}, it may be that we wish to identify multiple changepoints.  There are several standard techniques for extending single changepoint procedures to such settings, including binary segmentation and different versions of wild binary segmentation \citep{fryzlewicz2014wild,kovacs2020optimistic}.  Any of these approaches can be used in conjunction with the \texttt{MissInspect} algorithm to identify multiple changepoints in high-dimensional data streams in the presence of missingness.  The theoretical analysis of such a procedure would be technically involved, but would proceed along similar lines to that of \citet{wang2016highdimensional} for the case of fully-observed data. 

As always when handling missing data, the situation becomes much more complicated when the missingness and the data are not independent, i.e.~the Missing Completely At Random (MCAR) assumption does not hold.  In the worst case, the missingness may render the changepoint estimation problem impossible, for instance if no signal coordinate has observed data on both sides of the changepoint.  A less adversarial setting would be one in which all observations exceeding 1 are censored.  Thus, if the vector of mean change had positive entries in signal coordinates, we would expect to see fewer observations in these coordinates after the change.  The censoring would lead to (different) truncated Gaussian distributions before and after the change, but a difference in mean of these distributions would persist, so changepoint estimation may still be possible.  In general, careful and problem-specfic modelling of the dependence of the data and the missingness mechanism is recommended. 

Finally, we discuss settings of temporal and spatial dependence in the data.  In the former case, a natural model is to replace~\eqref{Eq:Xt} with 
\[
X_t = \mu_t + W_t, \quad \text{for } t=1,\ldots,n,
\]
where $\mu_1,\ldots,\mu_n$ satisfy~\eqref{Eq:SingleChangepoint} and where the noise vectors $(W_1,\ldots,W_n)$ form a mean-zero, stationary Gaussian process.  In this case, with row-homogeneous missingness independent of the data, the oracle projection direction remains $(\theta \circ \sqrt{\boldsymbol{q}})/\|\theta \circ \sqrt{\boldsymbol{q}}\|_2$, and the \texttt{MissInspect} methodology does not need to be altered.  On the other hand, if spatial dependence is introduced into~\eqref{Eq:SingleChangepoint} by replacing the identity covariance matrix there with a general covariance matrix $\Sigma$, then the oracle projection direction becomes proportional to $\Sigma^{-1}(\theta \circ \sqrt{\boldsymbol{q}})$.  If $\Sigma$ is unknown, then estimating $\Sigma^{-1}$ may represent a significant challenge, but it may be considerably simplified if our data stream satisfies additional structural assumptions.  For instance, if $\Sigma = \mathrm{diag}(\sigma_1^2,\ldots,\sigma_p^2)$, where $\sigma_1,\ldots,\sigma_p$ are unknown, then we can estimate these quantities robustly using, for example, the median absolute deviation of the marginal one-dimensional series \citep{hampel}.  As another example, if $\Sigma$ is Toeplitz with $\Sigma = (\rho^{|j-k|})_{j,k\in [p]}$ for some $\rho \in (-1,1)$, then $\Sigma^{-1}$ is tridiagonal, and its form can again be used to estimate $\rho$ \citep[][Lemma~12]{wang2016highdimensional}.

\section{Proof of main results}
\label{Sec:Proofs}

\subsection{Proof of Proposition~\ref{Prop:SineAngle}}
The proof of Proposition~\ref{Prop:SineAngle} requires the following result, which provides an entrywise control of the difference between the mean of the MissCUSUM transformation of $(X\circ\Omega,\Omega)$ conditional on $\Omega$ and its unconditional mean. 
\begin{proposition}
\label{Prop:Delta}
For $n\geq 2$, suppose $(X,\Omega) \sim P_{n,p,z,\theta,\sigma,\boldsymbol{q}}$ and let $A = \mathcal{T}\bigl(\mathbb{E}(X)\bigr) \in \mathbb{R}^{p \times (n-1)}$ and $A_\Omega = \mathcal{T}^{\mathrm{Miss}}\bigl(\mathbb{E}(X)\circ\Omega,\Omega\bigr) \in \mathbb{R}^{p \times (n-1)}$.  We have that 
\[
\mathbb{P} \bigl(  |(A_\Omega)_{j,t} - \sqrt{q_j}A_{j,t}| > 7\sqrt{6}|\theta_j| \sqrt{\log (kn)}\bigr) \leq \frac{4}{k^2n^2}
\]
for all $j \in [p]$ and $t \in [n-1]$.  Consequently, 
\[
    \mathbb{P} \bigl(  \|A_\Omega - (\diag \sqrt{\boldsymbol{q}})A\|_{\mathrm{F}} > 7\sqrt{6}\|\theta\|_2 \sqrt{n\log (kn)}\bigr) \leq \frac{4}{kn}.
\]
\end{proposition}

\begin{proof}
Let $\Delta := A_\Omega - \bigl(\diag\sqrt{\boldsymbol{q}}\bigr)A \in \mathbb{R}^{p \times (n-1)}$ have $(j,t)$th entry $\Delta_{j,t}$ and let $S := \{j \in [p]:\theta_j \neq 0\}$. Since $\Delta_{j,t}=0$ for $j\notin S$, it suffices to bound $|\Delta_{j,t}|$ for each $j\in S$ and $t\in [n-1]$.  Without loss of generality, we may assume that $\theta_j > 0$. 

First assume that $j$ satisfies $n\tau q_j \geq 24\log (kn)$. Let $ a_j := \bigl\lceil \frac{8\log(kn)}{3q_j}\bigr\rceil$. It follows that $a_j \leq n\tau/6$.   Define 
\[
\delta_{j,t} := \begin{cases} 
\frac{4\log(kn)}{tq_j} & \text{if $1\leq t\leq a_j-1$}\\
\sqrt\frac{6\log(kn)}{tq_j} & \text{if $a_j\leq t\leq n$.}
\end{cases}
\]
Recall that $L_{j,t} = \sum_{r=1}^t \omega_{j,r}$ and $R_{j,t} = \sum_{r=n-t+1}^n \omega_{j,r}$ for $j \in [p]$ and $t \in [n]$. We consider the event
\[
  \mathcal{A}_{j,t} := \biggl\{\max\biggl(\biggl|\frac{L_{j,t}}{tq_j} - 1\biggr|,\biggl|\frac{R_{j,t}}{tq_j} - 1\biggr|\biggr) \leq \delta_{j,t} \biggr\}.
\]
By Bernstein's inequality (Lemma~\ref{Lem:Bernstein}), we obtain that
\begin{align*}
\mathbb{P}\bigl(\mathcal{A}_{j,t}^{\mathrm{c}} \bigr)
&\leq 4 \exp\biggl\{-\frac{\delta_{j,t}^2tq_j}{2(1+\delta_{j,t}/3)}\biggr\} \\
&= 4 \exp\biggl\{-\frac{8\log^2(kn)}{tq_j + 4\log (kn)/3}\biggr\}\mathbbm{1}_{\{t < a_j\}} + 4 \exp\biggl\{-\frac{3\log(kn)}{1+\sqrt{2 \log(kn)/(3tq_j)}}\biggr\}\mathbbm{1}_{\{t\geq a_j\}} \\
&\leq 4 e^{-2\log(kn)} \leq \frac{4}{k^2n^2},
\end{align*}
where we have used the facts that $tq_j \leq 8\log (kn)/3$ for $t\leq a_j-1$ and $\sqrt{2\log(kn)} \leq \sqrt{3tq_j}/2$ for $t\geq a_j$ in the penultimate inequality. 

We will now bound $|\Delta_{j,t}|$ for different values of $t$ on $\mathcal{A}_{j,t}$. First, consider $a_j\leq t\leq z$.  Since $n-z \geq n\tau > a_j$, we have $\delta_{j, n-z}\leq \delta_{j, n\tau} \leq 1/2$.  We deduce from the definition of $A_{\Omega}$ that 
\begin{align}
(A_\Omega)_{j,t} &= \sqrt\frac{L_{j,t}}{R_{j,n-t}(L_{j,t}+R_{j,n-t})} R_{j,n-z}\theta_j \leq \sqrt\frac{(1+\delta_{j,t})tq_j}{(1-\delta_{j,n-t})^2(n-t)nq_j^2} (1+\delta_{j,n-z})(n-z)q_j\theta_j\nonumber\\
&\leq \sqrt{q_j}A_{j,t}\frac{\sqrt{1+\delta_{j,t}}(1+\delta_{j,n-z})}{1-\delta_{j,n-z}}.\label{Eq:tmp1}
\end{align}
It follows that for $j \in S$ and $a_j\leq t\leq z$, 
\begin{align*}
\Delta_{j,t} & \leq \sqrt{q_j}A_{j,t}\biggl\{\frac{\sqrt{1+\delta_{j,t}}(1+\delta_{j,n-z})}{1-\delta_{j,n-z}}-1\biggr\} \leq \sqrt{q_j} A_{j,t}\bigl\{(1+\delta_{j,t})(1+4\delta_{j,n-z})-1\bigr\} \\
&\leq \sqrt{q_j} A_{j,t}(3\delta_{j,t}+4\delta_{j,n-z}).
\end{align*}
By a similar calculation for deviations in the opposite direction, we have 
\[
\Delta_{j,t}\geq -\sqrt{q_j}A_{j,t}(\delta_{j,t}+4\delta_{j,n-z}/3).
\]
Thus, using the fact that $A_{j,t} \leq \theta_j\min\bigl(\sqrt{t},\sqrt{n-z}\bigr)$, we deduce that
\[
|\Delta_{j,t}|\leq \sqrt{q_j} A_{j,t}(3\delta_{j,t}+4\delta_{j,n-z}) \leq 7 \sqrt{q_j}\theta_j\min\bigl(\sqrt{t},\sqrt{n-z}\bigr) \max(\delta_{j,t}, \delta_{j,n-z}) \leq 7\theta_j\sqrt{6\log (kn)}.
\]
By symmetry, if $z < t \leq n-a_j$, we also have $|\Delta_{j,t}| \leq 7\theta_j\sqrt{6\log (kn)}$.

Next, if $t\leq a_j-1$, then we necessarily have $t\leq n\tau$. The calculation in~\eqref{Eq:tmp1} still applies, and we have
\[
(A_\Omega)_{j,t} \leq \sqrt{q_j}A_{j,t}\frac{\sqrt{1+\delta_{j,t}}(1+\delta_{j,n-z})}{1-\delta_{j,n-z}} \leq 3\sqrt{q_j}A_{j,t}\sqrt{1+\delta_{j,t}}.
\]
Hence, since $\mathrm{sgn}\bigl((A_\Omega)_{j,t}\bigr) = \mathrm{sgn}(A_{j,t})$, we have
\begin{align*}
|\Delta_{j,t}|\leq \max\bigl((A_\Omega)_{j,t}, \sqrt{q_j}A_{j,t}\bigr) &\leq A_{j,t}\sqrt{q_j} \max\bigl(3\sqrt{1+\delta_{j,t}}, 1\bigr)\\
&\leq \theta_j\sqrt{tq_j} \max\biggl\{3\sqrt\frac{tq_j+4\log(kn)}{tq_j}, 1\biggr\} \leq 8\theta_j\sqrt{\log (kn)}.
\end{align*}
A symmetric argument shows that $|\Delta_{j,t}|\leq 8\theta_j\sqrt{\log (kn)}$ for $n-a_j\leq t\leq n-1$.  Combining the above bounds on $|\Delta_{j,t}|$, we see that for $j$ satisfying $n\tau q_j \geq 24\log (kn)$ and all $t\in[n-1]$, we have that
\begin{equation}
\label{Eq:LargeProb}
\mathbb{P}\bigl(|\Delta_{j,t}| > 7\sqrt{6}\theta_j\sqrt{\log (kn)}\bigr) \leq \mathbb{P}(\mathcal{A}_{j,t}^{\mathrm{c}}) \leq \frac{4}{k^2n^2}.
\end{equation}

We now turn our attention to $j$ satisfying $n\tau q_j < 24\log (kn)$. If $q_j = 0$, then $\Delta_{j,t} = 0$. So we may assume $q_j>0$. Define
\[
\epsilon_j := \frac{24\log (kn)}{n\tau q_j},
\]
so that $\epsilon_j >1$.  For $j \in S$, consider the event
\[
  \mathcal{B}_{j} := \biggl\{ \max\biggl(\frac{L_{j,z}}{zq_j}, \frac{R_{j,n-z}}{(n-z)q_j}\biggr) \leq 1 + \epsilon_j \biggr\}.
\]
By Lemma~\ref{Lem:Bernstein} again, we have
\[
\mathbb{P}(\mathcal{B}_{j}^{\mathrm{c}}) \leq 2 \exp\biggl\{-\frac{\epsilon_j^2n\tau q_j}{2(1+\epsilon_j/3)}\biggr\} \leq 2e^{-9\log(kn)} = \frac{2}{(kn)^9}.
\]
On $\mathcal{B}_{j}$, we have
\[
(A_\Omega)_{j,t} \leq (A_\Omega)_{j,z} \leq \theta_j\sqrt{\min\{L_{j,z}, R_{j,n-z}\}}  \leq \theta_j \sqrt{(1+\epsilon_j)n\tau q_j} \leq \theta_j\sqrt{48\log (kn)}.
\]
On the other hand,
\[
\sqrt{q_j} A_{j,t} \leq \sqrt{q_j} A_{j,z} \leq \sqrt{q_j} \min\{\sqrt{z}, \sqrt{n-z}\}\theta_j = \theta_j \sqrt{n\tau q_j} \leq \theta_j \sqrt{24\log (kn)}.
\]
Consequently, when $n\tau q_j < 24\log (kn)$, we have 
\begin{equation}
\label{Eq:SmallProb}
\mathbb{P}\bigl(|\Delta_{j,t}| > \theta_j\sqrt{48\log(kn)}\bigr) \leq \mathbb{P}\bigl((A_\Omega)_{j,t}  > \theta_j\sqrt{48\log(kn)}\bigr) \leq \mathbb{P}(\mathcal{B}_j^{\mathrm{c}}) \leq \frac{2}{(kn)^9}.
\end{equation}
The first claim follows from~\eqref{Eq:LargeProb} and~\eqref{Eq:SmallProb}.  It now follows that
\[
\mathbb{P}\bigl(\|\Delta\|_{\mathrm{F}} > 7\sqrt{6}\|\theta\|_2\sqrt{n\log(kn)}\bigr) \leq \sum_{j\in S}\sum_{t=1}^{n-1} \mathbb{P}\bigl(|\Delta_{j,t}| > 7\sqrt{6}\theta_j\sqrt{\log(kn)}\bigr) \leq \frac{4}{kn},
\]
as desired.
\end{proof}

\begin{proof}[Proof of Proposition~\ref{Prop:SineAngle}]
Let $v \in \mathbb{S}^{p-1}$ denote the leading left singular vector of $A_\Omega$ and let $\sigma_1 \geq \sigma_2 \geq 0$ denote the two largest singular values of $A_\Omega$.  We start by controlling the angle between $\hat v$ and $v$. Write $\Delta:=A_{\Omega} - (\mathrm{diag}\sqrt{\boldsymbol{q}}) A \in \mathbb{R}^{p \times (n-1)}$ as in the proof of Proposition~\ref{Prop:Delta}.  Since $A = \theta\gamma^\top$, we have $(\mathrm{diag}\sqrt{\boldsymbol{q}}) A = (\theta\circ \sqrt{\boldsymbol{q}})\gamma^\top$. Hence, by Weyl's inequality \citep[e.g.][Corollary~IV.4.9]{StewartSun1990}, we obtain
\[
\sigma_1-\sigma_2 \geq \|\theta\|_{2,\boldsymbol{q}}\|\gamma\|_2 - 2\|\Delta\|_{\mathrm{op}} \geq \frac{n\tau \|\theta\|_{2,\boldsymbol{q}}}{4} - 2\|\Delta\|_{\mathrm{F}},
\]
where the final bound uses \citet[Lemma~3]{wang2016highdimensional}.  By Proposition~\ref{Prop:Delta}, there is an event~$\mathcal{A}$ with probability at least $1-4/(kn)$ such that  $\|\Delta\|_{\mathrm{F}} \leq 7\sqrt{6}\|\theta\|_2\sqrt{n\log(kn)}$. We may assume that $\sqrt{n}\tau \|\theta\|_2\geq 112\|\theta\|_{2,\boldsymbol{q}}\sqrt{6\log(kn)}$, since otherwise, the proposition is trivially true. With this assumption, we have on $\mathcal{A}$ that $\sigma_1-\sigma_2\geq n\tau\|\theta\|_{2,\boldsymbol{q}}/8$. Thus, by Lemma~\ref{Lem:2}, on the event $\mathcal{A}\cap \{\|T_\Omega - A_\Omega\|_\infty \leq \lambda n^{-1/2}\}$, we have that
\begin{equation}
\label{Eq:Angle1}
\sin \angle(\hat{v}, v) \leq \frac{4\lambda \sqrt{k}}{\sigma_1-\sigma_2} \leq \frac{32\lambda \sqrt{k}}{n\tau \|\theta\|_{2,\boldsymbol{q}}}.
\end{equation}
On the other hand, by \citet[Theorem~1.4]{Wang2016} (an extension of \citet[][Corollary~1]{YuWangSamworth2015}), on~$\mathcal{A}$, we also have that
\begin{equation}
\label{Eq:Angle2}
\sin\angle(v, \theta\circ \sqrt{\boldsymbol{q}}) \leq \frac{4\|\Delta\|_{\mathrm{op}}}{n\tau\|\theta\|_{2,\boldsymbol{q}}/4} \leq \frac{112\|\theta\|_{2}}{\tau\|\theta\|_{2,\boldsymbol{q}}}\sqrt\frac{6\log(kn)}{n}.
\end{equation}
By the triangle inequality, we deduce from~\eqref{Eq:Angle1} and~\eqref{Eq:Angle2} that on $\mathcal{A} \cap \{\|T_\Omega- A_\Omega\|_\infty \leq \lambda n^{-1/2}\}$,
\[
\sin \angle(\hat{v}, \theta\circ \sqrt{\boldsymbol{q}}) \leq \frac{32\lambda \sqrt{k}}{n\tau \|\theta\|_{2,\boldsymbol{q}}} + \frac{112\|\theta\|_{2}}{\tau\|\theta\|_{2,\boldsymbol{q}}}\sqrt\frac{6\log(kn)}{n}.
\]
The proposition follows on observing that
\begin{align*}
  \mathbb{P}(\mathcal{A}^{\mathrm{c}} \cup \{\|T_\Omega - A_\Omega\|_\infty > \lambda n^{-1/2}\}) &\leq \frac{4}{kn} + \sum_{j=1}^p \sum_{t=1}^{n-1} \mathbb{P}(|(T_{\Omega})_{j,t} - (A_\Omega)_{j,t}| > \lambda n^{-1/2})\\
  & \leq \frac{4}{kn} + pne^{-\lambda^2/(2n\sigma^2)} \leq \frac{6}{kn},
\end{align*}
where the penultimate inequality uses the fact that $(A_\Omega)_{j,t} - (T_{\Omega})_{j,t} \mid \Omega \sim N(0,\sigma^2)$ for all $t\in[n-1]$ and $j\in[p]$ such that $L_{j,t}R_{j,n-t}\neq 0$, and is equal to 0 when $L_{j,t}R_{j,n-t} = 0$.
\end{proof}

\subsection{Proof of Theorem~\ref{Thm:SlowRate}}
\begin{proof}[Proof of Theorem~\ref{Thm:SlowRate}]
Recall from Algorithm~\ref{Algo:MissInspectVariant} that $n_1 = n/2$, and for $\ell \in \{1,2\}$, let $\Omega^{(\ell)} \in \{0,1\}^{p \times n_1}$, $X^{(\ell)} \in \mathbb{R}^{p \times n_1}$ and $X_{\Omega}^{(\ell)} \in \mathbb{R}^{p \times n_1}$ denote the matrices formed from the $n_1$ odd columns (when $\ell=1$) and the $n_1$ even numbered columns (when $\ell=2$) of $\Omega$, $X$ and $X_\Omega = X\circ \Omega$ respectively. 
  For $\ell \in \{1,2\}$, let $T_{\Omega}^{(\ell)} := \mathcal{T}^{\mathrm{Miss}}(X_{\Omega}^{(\ell)},\Omega^{(\ell)}) \in \mathbb{R}^{p \times (n_1-1)}$.  By Proposition~\ref{Prop:SineAngle}, the output $\hat{v}$ of Algorithm~\ref{Algo:MissInspect} with inputs $T_{\Omega}^{(1)}$ and $\lambda$ satisfies
  \begin{equation}
    \label{Eq:SinAngle}
    \mathbb{P}\biggl\{ \sin \angle (\hat{v}, \theta \circ \sqrt{\boldsymbol{q}} ) > \frac{32\lambda \sqrt{k}}{n_1\tau\|\theta\|_{2,\boldsymbol{q}}} + \frac{112\|\theta\|_2}{\tau\|\theta\|_{2,\boldsymbol{q}}}\sqrt\frac{6\log(kn_1)}{n_1} \biggr\} \leq \frac{6}{kn_1} = \frac{12}{kn}.
  \end{equation}
  We can therefore find a universal constant $C' > 0$ such that whenever~\eqref{Eq:C'} holds, we have that the event $\mathcal{A} := \bigl\{\sin \angle (\hat{v}, \theta \circ \sqrt{\boldsymbol{q}} ) \leq 1/2\bigr\}$ has probability at least $1-12/(kn)$. 

Writing $\mu^{(2)} := \mathbb{E}(X^{(2)}) \in \mathbb{R}^{p \times n}$, let $A^{(2)} = \mathcal{T}(\mu^{(2)})$ and $A_\Omega^{(2)} = \mathcal{T}^{\mathrm{Miss}}(\mu^{(2)}\circ\Omega^{(2)},\Omega^{(2)})$.  Our main decomposition of interest here is
  \[
T_\Omega^{(2)} = (\mathrm{diag}\sqrt{\boldsymbol{q}})A^{(2)} + \Delta^{(2)} + E_\Omega^{(2)},
\]
where $\Delta^{(2)} := A_\Omega^{(2)} - (\mathrm{diag}\sqrt{\boldsymbol{q}})A^{(2)}$ and $E_\Omega^{(2)} := T_\Omega^{(2)} - A_\Omega^{(2)}$.  Since Algorithm~\ref{Algo:MissInspectVariant} remains the same if we replace $\hat{v}$ in Step~\ref{Step:7} with $-\hat{v}$, we may assume without loss of generality that $\hat{v}^\top(\theta\circ\sqrt{\boldsymbol{q}}) \geq 0$. Since $(\mathrm{diag}\sqrt{\boldsymbol{q}})A^{(2)} = (\theta \circ \sqrt{\boldsymbol{q}})\gamma^{(2)\top}$, where $\gamma^{(2)} = (\gamma_1^{(2)},\ldots,\gamma_{n_1-1}^{(2)})^\top \in \mathbb{R}^{n_1-1}$
\[
  \gamma_t^{(2)} := \left\{ \begin{array}{ll} \sqrt{\frac{t}{(n_1-t)n_1}}(n_1-z/2) & \mbox{if $t \leq z/2$,} \\
                        \sqrt{\frac{(n_1-t)}{tn_1}} (z/2) & \mbox{if $t > z/2$,} \end{array} \right.
\]
we have $\bigl(\hat{v}^\top(\diag\sqrt{\boldsymbol{q}})A^{(2)}\bigr)_{t} \geq 0$ for all $t\in[n_1-1]$. On the event $\mathcal{A}$, we have 
\begin{equation}
\label{Eq:Thm1tmp1}
\bigl(\hat{v}^\top(\diag\sqrt{\boldsymbol{q}})A^{(2)}\bigr)_{z/2} \geq \frac{\sqrt{3}}{2}\|\theta\|_{2,\boldsymbol{q}}\gamma^{(2)}_{z/2} \geq \frac{\sqrt{3}}{4}\|\theta\|_{2,\boldsymbol{q}}\sqrt{n\tau}.
\end{equation}
Observe that for every $t \in [n_1-1]$, we have
\[
  (\hat{v}^\top E_{\Omega}^{(2)})_t \mid \Omega^{(2)} \sim N\bigl(0,\sigma^2\|\hat{v}_{J_t}\|_2^2\bigr),
\]
where $J_t := \bigl\{j \in [p]: \min\bigl(\sum_{r=1}^t (\Omega^{(2)})_{j,r},\sum_{r=t+1}^{n_1} (\Omega^{(2)})_{j,r}\bigr) > 0\bigr\}$. Since $\|\hat{v}_{J_t}\|_2 \leq 1$, we deduce that $(\hat{v}^\top E_\Omega^{(2)})_{t}$ is stochastically dominated by $N(0,\sigma^2)$. Hence, together with the first conclusion of Proposition~\ref{Prop:Delta} and a union bound, there exists an event $\mathcal{B}$ with probability at least $1-4/(kn_1) - 1/n_1$ such that on $\mathcal{B}$ we have
\begin{equation}
\label{Eq:Thm1tmp2}
\max_{t\in [n_1-1]} |(\hat v^\top \Delta^{(2)})_t| \leq 7\sqrt{6}\|\theta\|_2\sqrt{\log(kn)} \quad \text{and} \quad \max_{t\in[n_1-1]} |(\hat v^\top E_\Omega^{(2)})_t| \leq 2\sigma \sqrt{\log n}.
\end{equation}
Combining~\eqref{Eq:Thm1tmp1} and~\eqref{Eq:Thm1tmp2}, and by increasing the universal constant $C' > 0$ if necessary, we have by~\eqref{Eq:C'} that on $\mathcal{A}\cap\mathcal{B}$, 
\begin{align*}
\bigl(\hat{v}^\top T_\Omega^{(2)}\bigr)_{z/2} &= \bigl(\hat{v}^\top(\diag\sqrt{\boldsymbol{q}})A^{(2)}\bigr)_{z/2} + (\hat{v}^\top E_{\Omega}^{(2)})_{z/2} + (\hat v^\top \Delta^{(2)})_{z/2}\\
&\geq  \max\Bigl\{0, \, \max_{t\in[n_1-1]} \bigl\{ -(\hat{v}^\top E_{\Omega}^{(2)})_{t} - (\hat v^\top \Delta^{(2)})_{t}\bigr\} \Bigr\} > \max_{t\in[n_1-1]} \bigl(-\hat{v}^\top T_\Omega^{(2)}\bigr)_t.
\end{align*}
In particular, on $\mathcal{A}\cap\mathcal{B}$, we have from the definition of $\hat z$ that $(\hat{v}^\top T_{\Omega}^{(2)})_{\hat{z}/2} \geq (\hat{v}^\top T_{\Omega}^{(2)})_{z/2} \geq 0$, so on this event we have the basic inequality
\begin{align}
  \label{Eq:ProjectedDecomp}
  \bigl(\hat{v}^\top (\mathrm{diag}\sqrt{\boldsymbol{q}})A^{(2)}\bigr)_{z/2} &- \bigl(\hat{v}^\top (\mathrm{diag}\sqrt{\boldsymbol{q}})A^{(2)}\bigr)_{\hat{z}/2}\nonumber\\
  &\qquad  \leq \bigl|(\hat{v}^\top E_\Omega^{(2)})_{z/2} - (\hat{v}^\top E_\Omega^{(2)})_{\hat{z}/2}\bigr| + \bigl|(\hat{v}^\top \Delta^{(2)})_{z/2} - (\hat{v}^\top \Delta^{(2)})_{\hat{z}/2}\bigr|.
\end{align}
By \citet[][Lemma~7]{wang2016highdimensional} on the event $\mathcal{A}\cap\mathcal{B}$, for every $t \in [n_1-1]$, we have
\begin{align}
\label{Eq:Ingredient0}
  \bigl(\hat{v}^\top (\mathrm{diag}\sqrt{\boldsymbol{q}})A^{(2)}\bigr)_{z/2} - \bigl(\hat{v}^\top (\mathrm{diag}\sqrt{\boldsymbol{q}})A^{(2)}\bigr)_{t} &= \bigl|\hat{v}^\top(\sqrt{\boldsymbol{q}} \circ \theta)\bigr|(\gamma_{z/2} - \gamma_t) \nonumber \\
  &\geq \frac{\sqrt{3}}{2}\|\theta\|_{2,\boldsymbol{q}} \cdot \frac{2}{3\sqrt{6}}\min\biggl(\frac{|z/2-t|}{\sqrt{n_1\tau}},\frac{\sqrt{n_1\tau}}{2}\biggr) \nonumber \\
  &= \frac{1}{3\sqrt{2}} \|\theta\|_{2,\boldsymbol{q}}\min\biggl(\frac{|z/2-t|}{\sqrt{n_1\tau}},\frac{\sqrt{n_1\tau}}{2}\biggr).
\end{align}
Combining~\eqref{Eq:ProjectedDecomp},~\eqref{Eq:Ingredient0} and \eqref{Eq:Thm1tmp2}, we then have on $\mathcal{A}\cap\mathcal{B}$ that, 
\begin{equation}
\label{Eq:Combined}
\frac{1}{3\sqrt{2}} \|\theta\|_{2,\boldsymbol{q}}\min\biggl(\frac{|\hat{z}-z|}{2\sqrt{n_1\tau}},\frac{\sqrt{n_1\tau}}{2}\biggr) \leq 4\sigma\sqrt{\log n} + 14\sqrt{6}\|\theta\|_2 \sqrt{\log(kn)}.
\end{equation}
For $C' \geq 84\sqrt{3}$, we have by~\eqref{Eq:C'} that
\[
\frac{24\sqrt{2}\sigma}{\|\theta\|_{2,\boldsymbol{q}}}\sqrt\frac{\log n}{n_1\tau} + \frac{168\sqrt{3}\|\theta\|_2}{\|\theta\|_{2,\boldsymbol{q}}}\sqrt\frac{\log(kn)}{n_1\tau} < \frac{2C'}{\tau}\sqrt\frac{\log(pn)}{n}\biggl(\frac{\sigma\sqrt{k}}{\|\theta\|_{2,\boldsymbol{q}}}+\frac{\|\theta\|_2}{\|\theta\|_{2,\boldsymbol{q}}}\biggr) \leq 1,
\]
which means that the minimum on the left-hand side of~\eqref{Eq:Combined} must be achieved by the first term. We therefore deduce with probability at least $\mathbb{P}(\mathcal{A}\cap\mathcal{B})\geq 1 -22/n$ that,
\[
\frac{|\hat{z}-z|}{n\tau} \leq \frac{24\sigma\sqrt{\log n}+84\sqrt{6}\|\theta\|_2 \sqrt{\log(kn)}}{\|\theta\|_{2,\boldsymbol{q}}\sqrt{n\tau}} \leq 84\sqrt{6}\biggl(\frac{\sigma}{\|\theta\|_{2,\boldsymbol{q}}} + \frac{\|\theta\|_2}{\|\theta\|_{2,\boldsymbol{q}}}\biggr)\sqrt\frac{\log(kn)}{n\tau},
\]
as desired.
\end{proof}

\subsection{Proof of Theorem~\ref{Thm:FastRate}}
The proof of Theorem~\ref{Thm:FastRate} will make use of the following propositions.

\begin{proposition}
\label{Prop:DeltaDiff}
Let $(X,\Omega) \sim P_{n,p,z,\theta,\sigma,\boldsymbol{q}}$, let $A = (A_{j,t}) = \mathcal{T}\bigl(\mathbb{E}(X)\bigr) \in \mathbb{R}^{p \times (n-1)}$ and let $A_\Omega = (A_\Omega)_{j,t} = \mathcal{T}^{\mathrm{Miss}}\bigl(\mathbb{E}(X)\circ\Omega,\Omega\bigr) \in \mathbb{R}^{p \times (n-1)}$.  Write $\Delta = (\Delta_{j,t}) = A_{\Omega} - (\mathrm{diag}\sqrt{\boldsymbol{q}}) A \in \mathbb{R}^{p \times (n-1)}$, fix $v = (v_1,\ldots,v_p)^\top \in\mathbb{S}^{p-1}$ and let $\tau:=n^{-1}\min\{z,n-z\}$.  For any given $\delta \in (0,1]$, if $n\tau \min_{j\in[p]}q_j\geq 60k\log(12p/\delta)$, then for $t$ satisfying $|z-t| \leq n\tau/50$, we have with probability at least $1-\delta$ that
  \[
    \bigl|(v^\top \Delta)_z-(v^\top \Delta)_t\bigr| \leq  \frac{2\|\theta\|_{2,\boldsymbol{q}}|z\!-\!t|}{9\sqrt{n\tau}}+\sqrt\frac{2|z\!-\!t|\sum_{j\in[p]} v_j^2\theta_j^2  \log(12/\delta)}{n\tau}+\frac{4\log(12p/\delta)}{3\sqrt{n\tau}}\max_{j\in[p]} \frac{|v_j\theta_j|}{q_j^{1/2}}.
  \]
\end{proposition}
\begin{proof}
Without loss of generality, we may assume that $t < z$.  For each $j\in[p]$, by two Taylor expansions, there exist $\xi_j, \tilde \xi_j \in [t,z]$ such that 
\begin{align*}
&A_{j,z} - A_{j,t} = \frac{(n-z)\theta_j}{\sqrt{n}}\biggl(\sqrt{\frac{z}{n-z}} - \sqrt{\frac{t}{n-t}}\biggr) 
 \\
 &=\theta_j(z-t)\frac{\sqrt{z^{-1}+(n-z)^{-1}}}{2} + \theta_j(z-t)^2\frac{n^{1/2}(n-z)(n - 4\xi_j)}{8\xi_j^{3/2}(n-\xi_j)^{5/2}}\\
 &=\theta_j(z-t)\frac{\sqrt{t^{-1}+(n-z)^{-1}}}{2}+\theta_j(z-t)^2\biggl\{\frac{n^{1/2}(n-z)(n - 4\xi_j)}{8\xi_j^{3/2}(n-\xi_j)^{5/2}} - \frac{1}{4\tilde\xi_j^2\sqrt{\tilde\xi_j^{-1}+(n-z)^{-1}}}\biggr\}.
\end{align*}
Similarly, by another two Taylor expansions, there exist random variables $\Xi_j, \tilde\Xi_j \in [L_{j,t},L_{j,z}]$ such that
\begin{align*}
(A_\Omega)_{j,z} - (A_\Omega)_{j,t} &= \theta_j(L_{j,z}-L_{j,t})\frac{\sqrt{L_{j,t}^{-1}+R_{j,n-z}^{-1}}}{2} \\
&\qquad + \theta_j(L_{j,z}-L_{j,t})^2\biggl\{\frac{N_j^{1/2}R_{j,n-z}(N_j-4\Xi_j)}{8\Xi_j^{3/2}(N_j-\Xi_j)^{5/2}} - \frac{1}{4\tilde\Xi_j^2\sqrt{\tilde\Xi_j^{-1}+R_{j,n-z}^{-1}}}\biggr\}.
\end{align*}
We write
\begin{align*}
D_{1,j}&:=\frac{\theta_j}{2}\sqrt\frac{t^{-1}+(n-z)^{-1}}{q_j} \bigl\{(L_{j,z}-L_{j,t})-q_j(z-t)\bigr\} \\
D_{2,j}&:=\frac{\theta_j(L_{j,z}-L_{j,t})}{2}\biggl\{\sqrt{L_{j,t}^{-1}+R_{j,n-z}^{-1}}-\sqrt\frac{t^{-1}+(n-z)^{-1}}{q_j}\biggr\} \\
D_{3,j}&:=|\theta_j|q_j^{1/2}(z-t)^2\biggl\{\frac{1}{2}\biggl(\frac{n}{\xi_j(n-\xi_j)}\biggr)^{3/2} +  \frac{1}{4\tilde\xi_j^{3/2}}\biggr\} \\
D_{4,j}&:= |\theta_j|(L_{j,z}-L_{j,t})^2\biggl\{\frac{1}{2}\biggl(\frac{N_j}{\Xi_j(N_j-\Xi_j)}\biggr)^{3/2}+\frac{1}{4\tilde\Xi_j^{3/2}}\biggr\}.
\end{align*}
We then have the bound
\begin{equation}
\bigl|(v^\top \Delta)_z-(v^\top \Delta)_t\bigr| \leq \biggl|\sum_{j=1}^p v_j D_{1,j}\biggr| + \biggl|\sum_{j=1}^p v_j D_{2,j}\biggr| +\sum_{j=1}^p |v_j|D_{3,j}+\sum_{j=1}^p |v_j|D_{4,j}.
\label{Eq:lem4tmp0}
\end{equation}
We control the four terms on the right-hand side  of~\eqref{Eq:lem4tmp0} separately. For the first term, setting $y:=\sqrt{2(z-t)\sum_{j\in[p]}v_j^2\theta_j^2\log(12/\delta)} + (1/3)\max_{j\in[p]}|v_j\theta_j|q_j^{-1/2}\log(12/\delta)$, we consider the event
\[
\mathcal{B}_t := \biggl\{\biggl|\sum_{j=1}^p\frac{v_j\theta_j}{q_j^{1/2}}\sum_{r=t+1}^z(\omega_{j,r}-q_j)\biggr| \leq y\biggr\}.
\]
Since $(\omega_{j,r})_{j \in [p],r\in(t,z]}$ are independent $\mathrm{Bern}(q_j)$ random variables, we have by Lemma~\ref{Lem:Bernstein} that  $\mathbb{P}(\mathcal{B}_t^{\mathrm{c}})\leq \delta/6$. On $\mathcal{B}_t$, we have that
\begin{equation}
\label{Eq:lem4tmp1}
\biggl|\sum_{j=1}^p v_j D_{1,j}\biggr| = \frac{1}{2}\sqrt{t^{-1}+(n-z)^{-1}}\biggl|\sum_{j=1}^p\frac{v_j\theta_j}{q_j^{1/2}}\sum_{r=t+1}^z(\omega_{j,r}-q_j)\biggr| \leq \frac{0.8y}{\sqrt{n\tau}}.
\end{equation}
For the second term on the right-hand side of~\eqref{Eq:lem4tmp0}, let $\mathcal{F}$ denote the $\sigma$-algebra generated by $(\omega_{j,r}: j\in[p], r\in[n], r\notin [t+1,z])$, then $L_{j,z}-L_{j,t}$ is independent of $\mathcal{F}$, whereas
\[
G_j:=\frac{v_j\theta_j}{2}\Bigl\{\sqrt{L_{j,t}^{-1}+R_{j,n-z}^{-1}}-q_j^{-1/2}\sqrt{t^{-1}+(n-z)^{-1}}\Bigr\}
\]
is measurable with respect to $\mathcal{F}$. We can therefore apply Lemma~\ref{Lem:Bernstein} conditional on $\mathcal{F}$ to obtain that there is a event $\mathcal{C}_t$ with $\mathbb{P}(\mathcal{C}_t^{\mathrm{c}}\mid \mathcal{F}) \leq \delta/6$ on which
\begin{align}
\label{Eq:D2jcond}
\biggl|\sum_{j=1}^p v_jD_{2,j}\biggr| &= \biggl|\sum_{j=1}^p G_j\sum_{r=t+1}^z \omega_{j,r} \biggr| \nonumber\\
&\leq (z-t) \biggl|\sum_{j=1}^p G_j q_j\biggr| + \sqrt{2(z-t)\log(12/\delta)\sum_{j=1}^p G_j^2 q_j} + \frac{1}{3}\max_{j\in[p]} |G_j| \log(12/\delta).
\end{align}
For $0\leq a < b \leq n$, define
\begin{equation}
\label{Eq:Hjab}
H_{j,(a,b)}:=\biggl|\frac{L_{j,b}-L_{j,a}}{(b-a)q_j}-1\biggr|.
\end{equation}
We consider the event
\[
\mathcal{A}_{j,t} := \biggl\{\max\bigl(H_{j,(0,t)}, H_{j,(0,z)}, H_{j,(z,n)}\bigr)\leq \frac{1}{5}\biggr\} \cap \biggl\{ H_{j,(t,z)} \leq \sqrt\frac{2\log(12p/\delta)}{(z-t)q_j}+ \frac{\log(12p/\delta)}{3(z-t)q_j}\biggr\}.
\]
By Lemma~\ref{Lem:Bernstein} again, we obtain that
\[
\mathbb{P}(\mathcal{A}_{j,t}^{\mathrm{c}}) \leq 6\exp\biggl\{-\frac{(1/5)^2(1-1/50)n\tau q_j}{2(1+1/15)}\biggr\} +\frac{\delta}{6p}\leq 6e^{-n\tau q_j/60} + \frac{\delta}{6p} \leq \frac{2\delta}{3p},
\]
where we used the assumption $n\tau q_j \geq 60\log(12p/\delta)$ in the final inequality. On $\mathcal{A}_{j,t}$, we have
\[
|G_j| \leq \frac{|v_j\theta_j|}{2q_j^{1/2}}(\sqrt{5/4}-1)\sqrt{t^{-1}+(n-z)^{-1}} \leq \frac{0.084|v_j\theta_j|}{(n\tau q_j)^{1/2}},
\]
where we have used the fact that $t\geq (49/50)n\tau$.
Combining the above inequality with~\eqref{Eq:D2jcond}, on $\cap_{j\in[p]}\mathcal{A}_{j,t} \cap \mathcal{C}_t$, we have by the Cauchy--Schwarz inequality that
\begin{equation}
\label{Eq:lem4tmp2}
\biggl|\sum_{j=1}^p v_j D_{2,j}\biggr| \leq 0.084\frac{\|\theta\|_{2,\boldsymbol{q}}(z-t) + y}{(n\tau)^{1/2}}.
\end{equation}
For the third and fourth terms on the right-hand side of~\eqref{Eq:lem4tmp0}, 
since $|z-t|\leq n\tau/50$, we have 
\begin{align*}
  D_{3,j} &\leq |\theta_j|q_j^{1/2}(z-t)^2 \biggl\{\frac{1}{2}\biggl(\frac{2}{\min(t,n-z)}\biggr)^{3/2} + \frac{1}{4t^{3/2}}\biggr\}
  \leq \frac{1.8 |\theta_j|q_j^{1/2}(z-t)^2}{(n\tau)^{3/2}}.
\end{align*}
Moreover, on $\mathcal{A}_{j,t}$,
\begin{align*}
D_{4,j}&\leq |\theta_j|  q_j^2(z-t)^2(1+H_{j,(t,z)})^2 \biggl\{\frac{1}{2}\biggl(\frac{2}{\min(L_{j,t}, R_{j,n-z})}\biggr)^{3/2} + \frac{1}{4L_{j,t}^{3/2}}\biggr\} \\
&\leq 
 \frac{2.4|\theta_j|q_j^{1/2}(z-t)^2}{(n\tau)^{3/2}}(1+H_{j,(t,z)})^2\\
  &\leq \frac{4.8|\theta_j|q_j^{1/2}(z-t)^2}{(n\tau)^{3/2}}+\frac{60|\theta_j|\log^2(12p/\delta)}{(n\tau q_j)^{3/2}}, 
\end{align*}
where the final step uses the fact that $(1+\sqrt{2a} + a/3)^2\leq 2+25a^2$ for any $a>0$. Therefore, on $\cap_{j\in[p]}\mathcal{A}_{j,t}$, since $|z-t|\leq n\tau/50$ and $n\tau \min_{j\in[p]}q_j \geq 60k\log(12p/\delta)$, we have by the Cauchy--Schwarz inequality again that
\begin{align}
  \label{Eq:lem4tmp3}
  \sum_{j=1}^p |v_j| (D_{3,j}+D_{4,j}) &\leq \frac{6.6\|\theta\|_{2,\boldsymbol{q}}(z-t)^2}{(n\tau)^{3/2}} + \frac{60\log^2(12p/\delta)}{(n\tau)^{3/2}}\sum_{j:\theta_j \neq 0}\frac{|v_j\theta_j|}{q_j^{3/2}}\nonumber\\
  &\leq \frac{0.132\|\theta\|_{2,\boldsymbol{q}}(z-t)}{(n\tau)^{1/2}} + \frac{\log(12p/\delta)}{(n\tau)^{1/2}}\max_{j:\theta_j \neq 0} \frac{|v_j\theta_j|}{q_j^{1/2}}.
\end{align}
Therefore, combining~\eqref{Eq:lem4tmp0}, \eqref{Eq:lem4tmp1}, \eqref{Eq:lem4tmp2} and \eqref{Eq:lem4tmp3}, we have on $\cap_{j\in[p]}\mathcal{A}_{j,t}\cap \mathcal{B}_{t} \cap \mathcal{C}_t$ that 
\begin{align*}
\bigl|(v^\top \Delta)_z-(v^\top \Delta)_t\bigr| &\leq \frac{y}{\sqrt{n\tau}} + \frac{2\|\theta\|_{2,\boldsymbol{q}}(z-t)}{9\sqrt{n\tau}}+ \frac{\log(12p/\delta)}{(n\tau)^{1/2}}\max_{j\in[p]} \frac{|v_j\theta_j|}{q_j^{1/2}}\\
&\hspace{-1.5cm}\leq \frac{2\|\theta\|_{2,\boldsymbol{q}}(z-t)}{9\sqrt{n\tau}}+\sqrt\frac{2(z-t)\sum_{j\in[p]} v_j^2\theta_j^2  \log(12/\delta)}{n\tau}+\frac{4\log(12p/\delta)}{3(n\tau)^{1/2}}\max_{j\in[p]} \frac{|v_j\theta_j|}{q_j^{1/2}}.
\end{align*}
Since $\sum_{j=1}^p \mathbb{P}(\mathcal{A}_{j,t}^{\mathrm{c}}) + \mathbb{P}(\mathcal{B}_t^{\mathrm{c}}) + \mathbb{P}(\mathcal{C}_t^{\mathrm{c}}) \leq \delta$, the proof is complete.
\end{proof}

\begin{proposition}
\label{Prop:EOmegaDiff}
Suppose that $\Omega = (\omega_{j,t})_{j \in [p],t \in [n]}$ and $W = (W_{j,t})_{j \in [p],t \in [n]}$ are independent, with $\omega_{j,t}\sim \mathrm{Bern}(q_j)$ independently and $q_j\in (0,1]$, and with $W_{j,t} \stackrel{\mathrm{iid}}{\sim} N(0,\sigma^2)$. Let $E_\Omega := \mathcal{T}^{\mathrm{Miss}}(W\circ\Omega,\Omega)$, let $z\in[n-1]$ and let $\tau:=n^{-1}\min\{z,n-z\}$. Suppose that $t \in [n-1]$ satisfies  $|z-t| \leq n\tau/2$.  For a fixed $v\in\mathbb{S}^{p-1}$, if $n\tau \min_{j\in[p]}q_j\geq 20\log(11p/\delta)$, then we have for any $\delta \in (0,1]$ that
\[
\mathbb{P}\Biggl\{\bigl|(v^\top E_{\Omega})_z - (v^\top E_{\Omega})_t\bigr| > 70\sigma \sqrt\frac{|z-t|\log(11/\delta) + \log^2(11/\delta)\max_{j\in[p]}v_j^2/q_j}{n\tau} \Biggr\} \leq \delta.
\]
\end{proposition}
\begin{proof} 
By symmetry, we may assume without loss of generality that $t < z$.  We note that $(E_{\Omega})_{j,z} - (E_{\Omega})_{j,t}$ is a centred normal random variable conditional on $\Omega$, so we start by looking at its conditional variance. By definition of $\mathcal{T}^{\mathrm{Miss}}$, we have
\begin{align}
  (E_{\Omega})_{j,z} - (E_{\Omega})_{j,t} &= \sqrt\frac{N_j}{L_{j,z}R_{j,n-z}}\biggl(\frac{L_{j,z}}{N_j}\sum_{r=1}^n W_{j,r}\omega_{j,r} - \sum_{r=1}^z  W_{j,r}\omega_{j,r}\biggr) \nonumber\\ 
  &\quad\qquad - \sqrt\frac{N_j}{L_{j,t}R_{j,n-t}} \biggl(\frac{L_{j,t}}{N_j}\sum_{r=1}^n W_{j,r}\omega_{j,r} - \sum_{r=1}^t  W_{j,r}\omega_{j,r}\biggr) \nonumber\\
  &= \sqrt\frac{N_j}{L_{j,z}R_{j,n-z}}\biggl(\frac{L_{j,z}-L_{j,t}}{N_j}\sum_{r=1}^n W_{j,r}\omega_{j,r} - \sum_{r=t+1}^z W_{j,r}\omega_{j,r}\biggr)\nonumber\\
  &\quad\qquad + \biggl(\sqrt\frac{N_j}{L_{j,z}R_{j,n-z}} - \sqrt\frac{N_j}{L_{j,t}R_{j,n-t}}\biggr)\biggl(\frac{L_{j,t}}{N_j}\sum_{r=1}^n W_{j,r}\omega_{j,r} - \sum_{r=1}^t W_{j,r}\omega_{j,r}\biggr).\label{Eq:VarEDiff}
\end{align}
Now, by the mean value theorem, there exists a random variable $\Xi_j \in [L_{j,t}, L_{j,z}]$ such that
\begin{align}
\label{Eq:MVT}
\biggl|\sqrt\frac{N_j}{L_{j,z}R_{j,n-z}} - \sqrt\frac{N_j}{L_{j,t}R_{j,n-t}}\biggr| &\leq (L_{j,z}-L_{j,t})\biggl|\frac{\Xi_j}{N_{j}}-\frac{1}{2}\biggr|\biggl(\frac{N_{j}}{\Xi_j(N_j-\Xi_j)}\biggr)^{3/2} \nonumber\\
&\leq \frac{\sqrt{2}(L_{j,z}-L_{j,t})}{\min(\Xi_j,N_j-\Xi_j)^{3/2}}.
\end{align}
Also, observe that
\begin{equation}
\label{Eq:EqualDiff}
\frac{L_{j,t}}{N_j}\sum_{r=1}^n W_{j,r}\omega_{j,r} - \sum_{r=1}^t W_{j,r}\omega_{j,r} = \sum_{r=t+1}^n W_{j,r}\omega_{j,r} - \frac{R_{j,n-t}}{N_j}\sum_{r=1}^n W_{j,r}\omega_{j,r}.
\end{equation}
Substituting~\eqref{Eq:MVT} and~\eqref{Eq:EqualDiff} into~\eqref{Eq:VarEDiff}, and observing that $\sum_{r=1}^nW_{j,r}\omega_{j,r}$ is positively correlated with each of $\sum_{r=t+1}^z W_{j,r}\omega_{j,r}$, $\sum_{r=1}^t W_{j,r}\omega_{j,r}$ and $\sum_{r=t+1}^n W_{j,r}\omega_{j,r}$, we have that
\begin{align}
\mathrm{Var}\bigl((E_{\Omega})_{j,z} - (E_{\Omega})_{j,t} \bigm| \Omega\bigr) &\leq \frac{2\sigma^2N_j}{L_{j,z}R_{j,n-z}}\biggl(\frac{(L_{j,z}-L_{j,t})^2}{N_j} + L_{j,z}-L_{j,t}\biggr) \nonumber\\
&\qquad +\frac{4\sigma^2(L_{j,z}-L_{j,t})^2}{\min(\Xi_j,N_j-\Xi_j)^3}\min\biggl(\frac{L_{j,t}^2}{N_j}+ L_{j,t}, R_{j,n-t}+\frac{R_{j,n-t}^2}{N_j}\biggr)\nonumber\\
& \leq 4\sigma^2(L_{j,z}-L_{j,t})\biggl(\frac{1}{L_{j,z}}+\frac{1}{R_{j,n-z}}\biggr) \nonumber\\
&\qquad + \frac{8\sigma^2(L_{j,z}-L_{j,t})^2}{\min\{L_{j,t}, R_{j,n-z}\}^2} \max\biggl(\frac{L_{j,z}}{L_{j,t}}, \frac{R_{j,n-t}}{R_{j,n-z}}\biggr).
\label{Eq:VarianceTBC}
\end{align}
Recalling the definition of $H_{j,(a,b)}$ from~\eqref{Eq:Hjab} in the proof of Proposition~\ref{Prop:DeltaDiff}, we consider the event
\[
\mathcal{A}_{j,t}:=\biggl\{\max\{H_{j,(0,z)}, H_{j,(0,t)}, H_{j,(z,n)}, H_{j,(t, n)}\} \leq \frac{1}{2}\biggr\}\\cap\biggl\{ H_{j,(t,z)} \leq \frac{n\tau}{ z-t} \biggr\}.
\]
By Bernstein's inequality (Lemma~\ref{Lem:Bernstein}), we obtain that
\[
\mathbb{P}(\mathcal{A}_{j,t}^{\mathrm{c}}) \leq 8\exp\biggl\{-\frac{(1/2)^2(n\tau/2)q_j}{2(1+1/6)}\biggr\} + \exp\biggl\{-\frac{(n\tau/(z-t))^2(z-t)q_j}{2\bigl(1+n\tau/\{3(z-t)\}\bigr)}\biggr\}\leq 9e^{-n\tau q_j/20},
\]
where we used the fact that $z-t\leq n\tau/2$ in the final inequality.  It therefore follows from~\eqref{Eq:VarianceTBC} that on the event $\mathcal{A}_{j,t}$, 
\begin{align*}
\mathrm{Var}\bigl((E_{\Omega})_{j,z} - (E_{\Omega})_{j,t} \bigm| \Omega\bigr) &\leq \frac{16\sigma^2}{n\tau q_j}(L_{j,z}-L_{j,t}) + \frac{768\sigma^2}{(n\tau q_j)^2}(L_{j,z}-L_{j,t})^2\leq \frac{1168\sigma^2}{n\tau q_j}(L_{j,z}-L_{j,t}).
\end{align*}
Hence, on $\cap_{j\in[p]}\mathcal{A}_{j,t}$, we have 
\[
\mathrm{Var}\bigl\{(v^\top E_{\Omega})_z - (v^\top E_{\Omega})_t \bigm|\Omega\bigr\}\leq \frac{1168\sigma^2}{n\tau} \sum_{j=1}^p \frac{v_j^2}{q_j}(L_{j,z}-L_{j,t})= \frac{1168\sigma^2}{n\tau} \sum_{j=1}^p \frac{v_j^2}{q_j}\sum_{r=t+1}^z \omega_{j,r}.
\]
Now, setting $y:=\sqrt{2(z-t)\sum_{j\in[p]}v_j^4q_j^{-1}\log(11/\delta)} +  (1/3)\max_{j\in[p]}v_j^2q_j^{-1}\log(11/\delta)$, consider the event
\[
\mathcal{B}_t := \biggl\{\sum_{j=1}^p \frac{v_j^2}{q_j}\sum_{r=t+1}^z \omega_{j,r} \leq z-t+y\biggr\}.
\]
We have by Bernstein's inequality (Lemma~\ref{Lem:Bernstein}) that $\mathbb{P}(\mathcal{B}_t^{\mathrm{c}})\leq \delta/11$. Noting that $\sum_{j\in[p]}v_j^4q_j^{-1}\leq \max_{j\in[p]}v_j^2q_j^{-1}$, and using the fact that $a+\sqrt{2ab}+b/3\leq 2(a+b)$ for any $a,b>0$, we have from the Gaussian tail bound that for every $u > 0$,
\begin{align*}
\mathbb{P}\Biggl\{\bigl|(v^\top E_{\Omega})_z &- (v^\top E_{\Omega})_t)\bigr| > 49u\sigma \sqrt\frac{z-t + \log(11/\delta)\max_{j\in[p]}v_j^2/q_j}{n\tau} \Biggr\} \\
&\leq e^{-u^2/2} + \sum_{j=1}^p\mathbb{P}(\mathcal{A}_{j,t}^{\mathrm{c}})+\mathbb{P}(\mathcal{B}_t^{\mathrm{c}})\leq e^{-u^2/2}+ 9\sum_{j=1}^p e^{-n\tau q_j/20} + \frac{\delta}{11}.
\end{align*}
The result follows by taking $u:=\sqrt{2\log(11/\delta)}$ and using the fact that $n\tau \min_{j\in[p]} q_j \geq 20\log(11p/\delta)$. 
\end{proof}

\begin{proof}[Proof of Theorem~\ref{Thm:FastRate}]
We write $z_1:=z/2$ and $n_1:=n/2$. Taking $C,C' > 0$ from Theorem~\ref{Thm:SlowRate}, we may assume that $c \in (0,1/50]$ is small enough that the hypothesis~\eqref{Eq:C'} of Theorem~\ref{Thm:SlowRate} is satisfied when $\rho\leq c$. Hence, by Theorem~\ref{Thm:SlowRate}, there is an event $\mathcal{E}$ with probability at least $1-22/n$ such that
\[
\frac{|\hat z - z|}{n\tau} \leq C\sqrt\frac{\log(kn)}{n\tau}\biggl(\frac{\sigma}{\|\theta\|_{2,\boldsymbol{q}}} + \frac{\|\theta\|_2}{\|\theta\|_{2,\boldsymbol{q}}}\biggr) \leq C\rho.
\]
By further reducing $c > 0$ if necessary, we may assume that on $\mathcal{E}$, and when $\rho\leq c$, we have $|\hat z- z|\leq n\tau/50$.

Let $A^{(2)}, \Delta^{(2)}$ and $E_{\Omega}^{(2)}$ be defined as in the proof of Theorem~\ref{Thm:SlowRate}. With $\hat{v} = (\hat{v}_1,\ldots,\hat{v}_p)^\top \in \mathbb{S}^{p-1}$ as defined in Algorithm~\ref{Algo:MissInspectVariant}, an inspection of the proof of Theorem~\ref{Thm:SlowRate} reveals that on $\mathcal{E}$, we also have for all $t\in[z_1-n_1\tau\rho, z_1+n_1\tau\rho]$ and $\rho \leq c$ that
\begin{equation}
\label{Eq:thm2tmp0}
\bigl(\hat{v}^\top (\mathrm{diag}\sqrt{\boldsymbol{q}})A^{(2)}\bigr)_{z_1} - \bigl(\hat{v}^\top (\mathrm{diag}\sqrt{\boldsymbol{q}})A^{(2)}\bigr)_{t}\geq \frac{|z_1-t| \|\theta\|_{2,{\boldsymbol{q}}}}{3\sqrt{2n_1\tau}}.
\end{equation}
Recall that $\hat{v}$ is measurable with respect to the $\sigma$-algebra generated by the odd-numbered time points, and that $\Delta^{(2)}$ and $E_{\Omega}^{(2)}$ are measurable with respect to the $\sigma$-algebra generated by the even-numbered time points.  By taking the universal constant $C_1 > 0$ in the statement of the theorem to be sufficiently large, we can ensure that the lower bounds on $n\tau\min_{j \in [p]} q_j$ in Propositions~\ref{Prop:DeltaDiff} and~\ref{Prop:EOmegaDiff} are satisfied.  It follows by these propositions that when $\rho \leq c$, for each $t\in[z_1-n_1\tau\rho, z_1+n_1\tau\rho]$, there is an event $\mathcal{A}_t$ of probability at least $1-n^{-2}$ on which both
\begin{align}
\bigl|(\hat v^\top \Delta^{(2)})_{z_1}-(\hat v^\top \Delta^{(2)})_t\bigr| &- \frac{2\|\theta\|_{2,\boldsymbol{q}}|z_1-t|}{9\sqrt{n_1\tau}} \nonumber\\
&\lesssim \sqrt\frac{|z_1-t|\sum_{j\in[p]} \hat v_j^2\theta_j^2  \log n}{n\tau}+
\frac{\log(pn)}{\sqrt{n\tau}}\max_{j\in[p]} \frac{|\hat v_j\theta_j|}{q_j^{1/2}},\label{Eq:thm2tmp1}\\
\bigl|(\hat v^\top E_\Omega^{(2)})_{z_1} -(\hat v^\top E_\Omega^{(2)})_{t}\bigr| &\lesssim \sqrt\frac{|z_1-t|\sigma^2\log n}{n\tau} + \frac{\sigma \log n}{\sqrt{n\tau}} \max_{j\in[p]} \frac{|\hat v_j|}{q_j^{1/2}}.\label{Eq:thm2tmp2}
\end{align}
Combining~\eqref{Eq:thm2tmp0}, \eqref{Eq:thm2tmp1}, \eqref{Eq:thm2tmp2} and the basic inequality as in~\eqref{Eq:ProjectedDecomp}, we have on the event $\mathcal{E}\cap \bigcap_{t\in[z/2-n_1\tau\rho,z/2+n_1\tau\rho]}\mathcal{A}_t$ and with $\rho \leq c$ that 
\begin{equation}
\frac{|\hat z-z| \|\theta\|_{2,{\boldsymbol{q}}}}{\sqrt{n\tau}} \lesssim
\sqrt\frac{|\hat z-z|(\sigma^2+\sum_{j\in[p]} \hat v_j^2\theta_j^2) \log n}{n\tau} 
+ \frac{\log(pn)}{\sqrt{n\tau}}\max_{j\in[p]} \frac{|\hat v_j\theta_j|}{q_j^{1/2}}
+ \frac{\sigma \log n}{\sqrt{n\tau}} \max_{j\in[p]} \frac{|\hat v_j|}{q_j^{1/2}}.\label{Eq:thm2tmp3}
\end{equation}
Define $v=(v_j)_{j\in[p]} \in \mathbb{R}^p$ such that $v_j := \theta_j q_j^{1/2}/\|\theta\|_{2,\boldsymbol{q}}$.  Then we can write 
\[
\hat v = \alpha v + \beta w,
\]
for some unit-length (random) vector $w=(w_j)_{j\in[p]}$ that is orthogonal to $v$ and some $\alpha, \beta \in \mathbb{R}$ such that $\alpha^2+\beta^2=1$.  Moreover, by inspecting the proof of Theorem~\ref{Thm:SlowRate}, we see that on $\mathcal{E}$, we have $|\beta| = \sin \angle (\hat{v},v) \leq \rho$.  Then from~\eqref{Eq:thm2tmp3}, we have on $\mathcal{E}\cap \bigcap_{t\in[z/2-n_1\tau\rho,z/2+n_1\tau\rho]}\mathcal{A}_t$ that
\begin{align*}
|\hat z- z| &\lesssim \frac{\sigma^2+\alpha^2\sum_{j\in[p]}v_j^2\theta_j^2+\beta^2\sum_{j\in[p]}w_j^2\theta_j^2}{\|\theta\|_{2,\boldsymbol{q}}^2}\log n \\
&\qquad + \frac{|\alpha|\max_{j\in[p]}|v_j\theta_j|q_j^{-1/2}+|\beta|\max_{j\in[p]}|w_j\theta_j|q_j^{-1/2}}{\|\theta\|_{2,\boldsymbol{q}}}\log(pn)\\
&\qquad + \frac{\sigma(|\alpha|\max_{j\in[p]}|v_j|q_j^{-1/2} + |\beta|\max_{j\in[p]}|w_j|q_j^{-1/2})}{\|\theta\|_{2,\boldsymbol{q}}}\log n\\
&\lesssim \frac{\sigma^2\log n}{\|\theta\|_{2,\boldsymbol{q}}^2} 
+ \frac{\|\theta\|_{4,\boldsymbol{q}}^4\log n}{\|\theta\|_{2,\boldsymbol{q}}^4} 
+ \frac{\rho^2\|\theta\|_\infty^2\log n}{\|\theta\|_{2,\boldsymbol{q}}^2}\\
&\qquad + \frac{\|\theta\|_\infty^2\log(pn)}{\|\theta\|_{2,\boldsymbol{q}}^2} 
+ \frac{\rho \max_{j\in[p]}|\theta_j|q_j^{-1/2}\log(pn)}{\|\theta\|_{2,\boldsymbol{q}}}\\
&\qquad + \frac{\|\theta\|_\infty\sigma\log n}{\|\theta\|_{2,\boldsymbol{q}}^2}+\frac{\rho\sigma\log n}{\|\theta\|_{2,\boldsymbol{q}}\min_{j\in[p]}q_j^{1/2}}\\
&\lesssim \frac{\bigl(\sigma^2+\|\theta\|_\infty^2\bigr)\log (pn)}{\|\theta\|_{2,\boldsymbol{q}}^2} + \frac{\rho(\sigma+\|\theta\|_\infty)\log(pn)}{\|\theta\|_{2,\boldsymbol{q}}\min_{j\in[p]}q_j^{1/2}} \lesssim \frac{\bigl(\sigma^2+\|\theta\|_2^2\bigr)\log (pn)}{\|\theta\|_{2,\boldsymbol{q}}^2},
\end{align*}
where the final bound uses the definition of $\rho$ and the fact that $n\tau^2 \min_{j\in[p]}q_j \geq C_1 k\log(pn)$. 
The desired result follows since $\mathbb{P}\bigl(\mathcal{E}\cap \bigcap_{t\in[z/2-n_1\tau\rho,z/2+n_1\tau\rho]}\mathcal{A}_t\bigr) \geq 1 - 22/n - (2n_1\tau\rho+1)/n^2 \geq 1-23/n$.
\end{proof}

\subsection{Proof of Theorem~\ref{Thm:LowerBound}}
\begin{proof}[Proof of Theorem~\ref{Thm:LowerBound}]
For notational simplicity, we abbreviate $P_{n,p,z,\theta,\sigma,\boldsymbol{q}}$ as $P_z$ in this proof, with corresponding expectation operator $E_z$. For any $1\leq z_1 < z_2 \leq n-1$, by Le Cam's two point testing lemma \citep[e.g.][Lemma~1]{Yu1997}, we have that
\begin{equation}
\label{Eq:LeCam}
\inf_{\tilde z \in \tilde{\mathcal{Z}}}\max_{z \in [n-1]} E_z |\tilde z - z| \geq \frac{1}{2}|z_1-z_2|\bigl\{1 - d_{\mathrm{TV}}(P_{z_1}, P_{z_2})\bigr\}.
\end{equation}
By Pinsker's inequality \citep[e.g.][Lemma~15.2]{Wainwright2019}, we have
\begin{align*}
2d_{\mathrm{TV}}^2(P_{z_1},P_{z_2}) \leq \mathrm{KL}(P_{z_1}\,||\,P_{z_2}) &= \mathbb{E}_{P_{z_1}}\biggl[\mathbb{E}_{P_{z_1}}\biggl\{\log\biggl(\frac{dP_{z_1}}{dP_{z_2}}(X,\Omega)\biggr)\biggm|\Omega\biggr\}\biggr]\\
&=\sum_{j=1}^p\sum_{t=z_1+1}^{z_2} \mathbb{E}_{P_{z_1}}\frac{\theta_j^2\omega_{j,t}}{2\sigma^2} = \frac{(z_2-z_1)\|\theta\|_{2,\boldsymbol{q}}^2}{2\sigma^2}.
\end{align*}
Choosing $z_2-z_1=\min\bigl\{\lfloor \sigma^2/\|\theta\|_{2,\boldsymbol{q}}^2\rfloor, n-2\bigr\}$, we have $d_{\mathrm{TV}}(P_{z_1},P_{z_2})\leq 1/2$ and consequently if $\sigma^2\geq \|\theta\|_{2,\boldsymbol{q}}^2$, then by~\eqref{Eq:LeCam},
\begin{equation}
\label{Eq:LowerBound1}
\inf_{\tilde z \in \tilde{\mathcal{Z}}}\max_{z \in [n-1]} E_z |\tilde z - z| \geq \frac{1}{4}\min\biggl\{\biggl\lfloor\frac{\sigma^2}{\|\theta\|_{2,\boldsymbol{q}}^2}\biggr\rfloor, n-2\biggr\} \geq \frac{1}{8}\min\biggl\{\frac{\sigma^2}{\|\theta\|_{2,\boldsymbol{q}}^2}, n\biggr\}.
\end{equation}
On the other hand, if $\|\theta\|_\infty^2 \geq 2M^2 \|\theta\|_{2,\boldsymbol{q}}^2$, then 
\[
\sum_{j:\theta_j\neq 0} q_j \leq \frac{\|\theta\|_{2,\boldsymbol{q}}^2}{\min_{j:\theta_j \neq 0} \theta_j^2} \leq \frac{M^2\|\theta\|_{2,\boldsymbol{q}}^2}{\|\theta\|_\infty^2} \leq 1/2.
\]
Define $\mathcal{S}:= \{j \in [p]:\theta_j \neq 0\}$ and
\[
\mathcal{A} := \bigl\{\bigl((x_{j,t})_{j \in [p],t \in [n]},(\omega_{j,t})_{j\in [p],t \in [n]}\bigr): \omega_{j,t} = 0 \text{ whenever } j \in \mathcal{S} \text{ and } z_1+1 \leq t \leq z_2\bigr\}.
\]
Then the distributions of $P_{z_1}$ given $\mathcal{A}$ and $P_{z_2}$ given $\mathcal{A}$ are identical.  Moreover, $P_{z_1}(\mathcal{A}) = P_{z_2}(\mathcal{A})$.  Thus, for any Borel measurable subset $\mathcal{B}$ of $\mathbb{R}^{n \times p} \times \{0,1\}^{n \times p}$, we have
\[
|P_{z_1}(\mathcal{B}) - P_{z_2}(\mathcal{B})| = \bigl|P_{z_1}(\mathcal{B}\mid \mathcal{A}^{\mathrm{c}}) - P_{z_2}(\mathcal{B}\mid \mathcal{A}^{\mathrm{c}})\bigr| P_{z_1}(\mathcal{A}^{\mathrm{c}}) \leq P_{z_1}(\mathcal{A}^{\mathrm{c}}).
\]
Hence, using the fact that $1-x\geq e^{-2x\log 2}$ for $x \in [0,1/2]$, we have
\begin{align*}
1-d_{\mathrm{TV}}(P_{z_1},P_{z_2}) \geq P_{z_1}(\mathcal{A}) = \prod_{j:\theta_j\neq 0}(1-q_j)^{z_2-z_1}\geq \exp\biggl\{-2(\log 2)(z_2-z_1)\sum_{j:\theta_j\neq 0} q_j\biggr\}.
\end{align*}
Choosing $z_2-z_1=\min\bigl\{\lceil (2\sum_{j:\theta_j\neq 0}q_j)^{-1}\rceil, n-2\bigr\}$, we have $1-d_{\mathrm{TV}}(P_{z_1},P_{z_2})\geq 1/4$, and consequently, on combining with~\eqref{Eq:LeCam} we obtain that
\begin{equation}
\label{Eq:LowerBound2}
\inf_{\tilde z \in \tilde{\mathcal{Z}}}\max_{z\in [n-1]} E_z|\tilde z - z| \geq \frac{1}{8}\min\biggl\{\frac{1}{2\sum_{j:\theta_j\neq 0} q_j}, n-2\biggr\} \geq \frac{1}{16}\min\biggl\{\frac{\|\theta\|_\infty^2}{M^2 \|\theta\|_{2,\boldsymbol{q}}^2}, n\biggr\}.
\end{equation}
By combining~\eqref{Eq:LowerBound1} and~\eqref{Eq:LowerBound2}, and considering the three possible cases of (i) $\sigma^2 \geq \|\theta\|_{2,\boldsymbol{q}}^2 > \|\theta\|_\infty^2/(2M^2)$, (ii) $\|\theta\|_\infty^2/(2M^2)\geq \|\theta\|_{2,\boldsymbol{q}}^2 >\sigma^2$ and (iii) $\min\{\sigma^2, \|\theta\|_\infty^2/(2M^2)\}\geq \|\theta\|_{2,\boldsymbol{q}}^2$, we have
\begin{align*}
\inf_{\tilde z \in \tilde{\mathcal{Z}}}\max_{z \in [n-1]} E_z |\tilde z - z| &\geq \frac{1}{8}\max \biggl\{\frac{\sigma^2}{\|\theta\|_{2,\boldsymbol{q}}^2} \wedge \frac{n}{2},\frac{\|\theta\|_\infty^2}{2M^2 \|\theta\|_{2,\boldsymbol{q}}^2} \wedge \frac{n}{2}\biggr\} \\
&\geq \frac{1}{8}\min\biggl(\max\biggl\{\frac{\sigma^2}{\|\theta\|_{2,\boldsymbol{q}}^2},\frac{\|\theta\|_\infty^2}{2M^2 \|\theta\|_{2,\boldsymbol{q}}^2}\biggr\},\frac{n}{2}\biggr) \\
&\geq \frac{1}{16}\min\biggl(\frac{\sigma^2}{\|\theta\|_{2,\boldsymbol{q}}^2} + \frac{\|\theta\|_\infty^2}{2M^2 \|\theta\|_{2,\boldsymbol{q}}^2},n\biggr),
\end{align*}
as required.
\end{proof}

\appendix

\section{Auxiliary lemmas and proofs}
\label{App:1}
Define the soft-thresholding function $\mathrm{soft}: \mathbb{R}^{p}\times [0,\infty) \rightarrow \mathbb{R}^p$ such that for $v = (v_1,\ldots,v_p)^\top \in\mathbb{R}^p$, we have $\bigl(\mathrm{soft}(v, \lambda)\bigr)_j = \mathrm{sgn}(v_j)\max\{|v_j| - \lambda, 0\}$ for $j\in [p]$. 
\begin{lemma}
\label{Lem:1}
Let $M\in\mathbb{R}^{p\times n}$ and let
\[
(v_*,w_*) \in \argmax_{(v,w)\in\mathbb{B}^p \times \mathbb{B}^n} \bigl\{\langle M, vw^\top\rangle - \lambda\|v\|_1\bigr\}.
\]
If $\|M\|_{2 \rightarrow \infty} < \lambda$, then $v_* = 0$; if $\|M\|_{2 \rightarrow \infty} > \lambda$, then
\begin{equation}
\label{Eq:vstarwstar}
v_* = \frac{\mathrm{soft}(Mw_*, \lambda)}{\|\mathrm{soft}(Mw_*, \lambda)\|_2} \quad \mathrm{and} \quad w_* = \frac{M^\top v_*}{\|M^\top v_*\|_2}.
\end{equation}
Finally, if $\|M\|_{2 \rightarrow \infty} = \lambda$, then either $v_* = 0$ or both $w_* \in \{w \in \mathbb{R}^n:\|w\|_2= 1,\|Mw\|_\infty = \lambda\}$ and
\[
  \mathrm{sgn}\bigl(v_{*,j}\bigr) = \mathrm{sgn}\bigl((Mw_*)_j\bigr)\mathbbm{1}_{\{|(Mw_*)_j| = \lambda\}}
\]
for every $j \in [n]$, where $v_* = (v_{*,1},\ldots,v_{*,p})$.  
\end{lemma}
\begin{proof}
  For $(v,w)\in\mathbb{R}^p \times \mathbb{R}^n$ with $\|v\|_2 \leq 1, \|w\|_2 \leq 1$, we write
  \[
    f(v,w) := \langle M, vw^\top\rangle - \lambda\|v\|_1
  \]
  for our objective function.  We first note that maximisers exist since $f$ is concave and the constraint set is convex and compact.  Moreover, for $(v,w)\in \mathbb{B}^p \times \mathbb{B}^n$, we have 
  \[
    f(v,w) = v^\top Mw - \lambda\|v\|_1 \leq (\|Mw\|_\infty - \lambda)\|v\|_1 =: g(v,w).
\]
If $\|M\|_{2 \rightarrow \infty} < \lambda$, then $f(v,w) \leq g(v,w) \leq 0$, with both equalities holding if and only if $v = 0$, and we deduce that $v_* = 0$.  If $\|M\|_{2 \rightarrow \infty} = \lambda$, then again $f(v,w) \leq g(v,w) \leq 0$ with both equalities holding if and only if either $v=0$, or $v^\top Mw = \|Mw\|_\infty = \lambda$; the latter case yields the constraints on $w_*$ and $v_*$ given in the statement.  Finally, we consider the case where $\|M\|_{2 \rightarrow \infty} > \lambda$.  We can find $(v_0,w_0)\in\mathbb{S}^{p-1}\times\mathbb{S}^{n-1}$ such that $\|Mw_0\|_\infty = \|M\|_{2\to\infty}$ and $v_0^\top Mw_0 = \|v_0\|_1\|Mw_0\|_\infty$, and consequently $f(v_*,w_*) \geq f(v_0,w_0)=g(v_0,w_0)>0$. In particular, we may assume that $M^\top v_*\neq 0$ in the remainder of the proof.
Define the Lagrangian $\mathcal{L}: \mathbb{R}^{p}\times \mathbb{R}^{n}\times [0,\infty) \times [0,\infty) \to \mathbb{R}$ by
\[
    \mathcal{L}(v,w,\alpha,\beta) := \langle M, vw^\top \rangle  - \lambda \|v\|_1 - \alpha\bigl(\|v\|^2_2-1\bigr) - \beta\bigl(\|w\|^2_2-1\bigr).
\]
By the Karush--Kuhn--Tucker conditions, we have 
\begin{align*}
M^\top v_* - 2\beta w_* &= 0\\
Mw_* - \lambda\eta - 2\alpha v_* &= 0,
\end{align*}
where $\eta = (\eta_1,\ldots,\eta_p)\in[-1,1]^p$ satisfies $\eta_j = \mathrm{sgn}\bigl((v_*)_j\bigr)$ if $(v_*)_j \neq 0$. Therefore, we have $w_* \propto M^\top v_*$ and $v_* \propto \mathrm{soft}(Mw_*, \lambda)$, as desired, since $M^\top v_* \neq 0$.
\end{proof}

\begin{lemma}
\label{Lem:2} 
Suppose that $A, T \in \mathbb{R}^{p\times n}$ satisfy $\|T-A\|_\infty \leq \lambda n^{-1/2}$ for some $\lambda \geq 0$.  Suppose further that $v \in \mathbb{S}^{p-1}(k)$ and $w \in \mathbb{S}^{n-1}$ are respectively the leading left and right singular vectors of $A$, and that
\[
(\hat{v}, \hat{w}) \in \argmax_{(\tilde v,\tilde w)\in\mathbb{S}^{p-1}\times \mathbb{S}^{n-1}} \bigl\{\langle T, \tilde v\tilde w^\top \rangle - \lambda\|\tilde v\|_1\bigr\}.
\]
Let $\delta > 0$ denote the difference between the first and second singular values of $A$. Then
\[
\sin \angle\bigl(\hat{v}, v\bigr) \leq \frac{4\lambda \sqrt{k}}{\delta}.
\]
\end{lemma}

\begin{proof}
Let $S := \{ j \in [p]:v_j \neq 0\}$.  By Lemma~2 in the supplementary material of \citet{wang2016highdimensional}, we have
\begin{align}
\label{Eq:Decomp}
    \frac{\delta}{2}\|\hat{v}\hat{w}^\top - vw^\top\|_{\mathrm{F}}^2 &\leq \langle A, vw^\top - \hat{v}\hat{w}^\top \rangle = \langle T, vw^\top - \hat{v}\hat{w}^\top \rangle + \langle A-T, vw^\top - \hat{v}\hat{w}^\top \rangle \nonumber \\
                                                                     &\leq \lambda(\|v\|_1 - \|\hat v\|_1) + \lambda n^{-1/2}\|\hat v\hat w^\top - vw^\top\|_1 \nonumber \\
                                                                     &= \lambda \biggl\{\|v_S\|_1 - \|\hat{v}_S\|_1+ n^{-1/2}\|\hat{v}_S\hat{w}^\top - v_Sw^\top\|_1 + \|\hat{v}_{S^{\mathrm{c}}}\|_1(\|\hat{w}\|_1n^{-1/2}-1)\biggr\} \nonumber \\
  &\leq \lambda \bigl(\|v_S\|_1 - \|\hat{v}_S\|_1+ \sqrt{k}\|\hat{v}\hat{w}^\top - vw^\top\|_{\mathrm{F}}\bigr).
\end{align}
Moreover, writing $w_0 := (\hat{w}+w)/2$ and $\Delta := w-w_0 = (w - \hat{w})/2$, we have
\begin{align}
\label{Eq:Decomp1}
\|\hat{v}\hat{w}^\top - vw^\top\|_{\mathrm{F}}^2 &= \|\hat{v}(w_0 - \Delta)^\top - v(w_0 + \Delta)^\top\|_{\mathrm{F}}^2 = \|(\hat{v} - v)w_0^\top\|_{\mathrm{F}}^2 + \|(\hat{v}+v)\Delta^\top\|_{\mathrm{F}}^2 \nonumber \\
& = \|w_0\|_2^2 \|\hat{v}-v\|_2^2 + \|\Delta\|_2^2\|\hat{v}+v\|_2^2 \nonumber \\
& \geq (\|w_0\|_2^2+\|\Delta\|_2^2)\min(\|\hat{v}-v\|_2^2, \|\hat{v}+v\|_2^2) \nonumber \\
& \geq 2(1 - |\hat{v}^\top v|) \geq 1-(\hat{v}^\top v)^2 = \sin^2\angle(\hat{v},v),
\end{align}
where the penultimate step uses the fact that $\|w_0\|_2^2+\|\Delta\|_2^2 = 1$.  It follows that
\begin{align}
\label{Eq:Decomp2}
  \|v_S\|_1 - \|\hat{v}_S\|_1 \leq \min(\|\hat{v}_S - v_S\|_1, \|\hat{v}_S + v_S\|_1) &\leq \sqrt{k}\min(\|\hat{v} - v\|_2, \|\hat{v} + v\|_2) \nonumber \\
  &\leq \sqrt{k}\|\hat{v}\hat{w}^\top - vw^\top\|_{\mathrm{F}}
\end{align}
Substituting~\eqref{Eq:Decomp1} and~\eqref{Eq:Decomp2} into~\eqref{Eq:Decomp}, we conclude that
\begin{equation}
\begin{aligned}
    \sin \angle(\hat{v},v) \leq \|\hat{v}\hat{w}^\top - vw^\top\|_{\mathrm{F}} \leq \frac{4\lambda \sqrt{k}}{\delta},
    \end{aligned}
  \end{equation}
as required.
\end{proof}

We state below a version of the Bernstein's inequality that is convenient to apply in our setting. 
\begin{lemma}
\label{Lem:Bernstein}
If $X_1,\ldots,X_n$ are independent with $X_i\sim \mathrm{Bern}(q_i)$ for $q_i\in(0,1)$. Let $a=(a_i)_{i\in[n]}\in\mathbb{R}^n$ and define $\|a\|_{2,\boldsymbol{q}}:=\bigl(\sum_{i\in[n]}a_i^2q_i\bigr)^{1/2}$. Writing $S:=\sum_{i\in[n]}a_i(X_i-q_i)$, we have for any $y>0$ and $\delta \in (0,1)$ that
\[
\mathbb{P}(S \geq y) \leq \exp\biggl(-\frac{1}{2}\frac{y^2}{\|a\|_{2,\boldsymbol{q}}^2+\|a\|_\infty y/3}\biggr)
\]
and
\[
\mathbb{P}\biggl(S \geq 2^{1/2}\|a\|_{2,\boldsymbol{q}}\log^{1/2}(1/\delta) + \frac{\|a\|_\infty}{3}\log(1/\delta)\biggr) \leq \delta.
\]
In particular, if $Y\sim \mathrm{Bin}(n,q)$ for $n\in\mathbb{N}$ and $q\in(0,1)$ and $H:= Y/(nq) - 1$, then for any $u>0$ and $\delta \in (0,1)$ we have
\[
\mathbb{P}(H \geq u) \leq \exp\biggl(-\frac{1}{2}\frac{nqu^2}{1+u/3}\biggr)
\]
and
\[
\mathbb{P}\biggl(H \geq \sqrt\frac{2\log(1/\delta)}{nq} + \frac{\log(1/\delta)}{3nq}\biggr\}\biggr) \leq \delta.
\]
Moreover, the same conclusions hold with $-S$ and $-H$ replacing $S$ and $H$ respectively above. 
\end{lemma}
\begin{proof}
Writing $Y_i:=a_i(X_i-q_i)$ for $i \in [n]$, we have for any positive integer $r\geq 2$ that
\[
\mathbb{E}|Y_i|^r = a_i^r\{q_i(1-q_i)^r+(1-q_i)q_i^r\} \leq a_i^r q_i(1-q_i).
\]
Consequently, 
\[
\sum_{i=1}^n \mathbb{E}|Y_i|^r \leq \sum_{i=1}^n a_i^r q_i(1-q_i) \leq \sum_{i=1}^n \frac{r!}{2} 3^{-(r-2)} a_i^rq_i \leq \frac{r!}{2} \biggl(\frac{\|a\|_\infty}{3}\biggr)^{r-2}\|a\|_{2,\boldsymbol{q}}^2.
\]
Hence, the first two conclusions follows from  \citet[(2.10) and Theorem~2.10]{boucheron2013concentration}. The final two conclusions follows from the first two by setting $a = (1,\ldots,1)^\top \in\mathbb{R}^n$ and $y=nqu$.
\end{proof}

\noindent \textbf{Acknowledgements}: BF was supported by a Knox studentship from Trinity College, Cambridge and EOX funding from Ecole polytechnique; TW was supported by EPSRC grant EP/T02772X/1; RJS was supported by EPSRC grants EP/P031447/1 and EP/N031938.

\bibliographystyle{custom}
\bibliography{biblio}

\begin{thebibliography}{44}
\providecommand{\natexlab}[1]{#1}
\providecommand{\url}[1]{\texttt{#1}}
\providecommand{\urlprefix}{URL }
\providecommand{\eprint}[2][]{\url{#2}}

\bibitem[{Aston and Kirch(2013)}]{aston2012}
Aston, J. and Kirch, C. (2013) Evaluating stationarity via change-point
  alternatives with applications to fMRI data. \emph{Ann. Appl. Statist.},
  \textbf{6}, 1906--1948.

\bibitem[{Bai(2010)}]{Bai2010}
Bai, J. (2010) Common breaks in means and variances for panel data. \emph{J.
  Econom.}, \textbf{157}, 78--92.

\bibitem[{Boucheron, Lugosi and Massart(2013)}]{boucheron2013concentration}
Boucheron, S., Lugosi, G. and Massart, P. (2013) \emph{Concentration
  Inequalities: A Nonasymptotic Theory of Independence}. Oxford University
  Press, Oxford.

\bibitem[{Chan(2017)}]{chan2017optimal}
Chan, H.~P. (2017) Optimal sequential detection in multi-stream data.
  \emph{Ann. Statist.}, \textbf{45}, 2736--2763.

\bibitem[{Chan and Walther(2015)}]{chan2015optimal}
Chan, H.~P. and Walther, G. (2015) Optimal detection of multi-sample aligned
  sparse signals. \emph{Ann. Statist.}, \textbf{43}, 1865--1895.

\bibitem[{Chen and Gupta(1997)}]{chen1997testing}
Chen, J. and Gupta, A.~K. (1997) Testing and locating variance changepoints
  with application to stock prices. \emph{J. Amer. Statist. Assoc.},
  \textbf{92}, 739--747.

\bibitem[{Chen, Wang and Samworth(2021)}]{chen2020highdimensional}
Chen, Y., Wang, T. and Samworth, R.~J. (2021) High-dimensional, multiscale
  online changepoint detection. \emph{J. Roy. Statist. Soc., Ser. B}, to
  appear.

\bibitem[{Cho(2016)}]{Cho_2016}
Cho, H. (2016) Change-point detection in panel data via double CUSUM statistic.
  \emph{Electron. J. of Statist.}, \textbf{10}, 2000--2038.

\bibitem[{Cho and Fryzlewicz(2014)}]{Cho_2014}
Cho, H. and Fryzlewicz, P. (2014) Multiple-change-point detection for high
  dimensional time series via sparsified binary segmentation. \emph{J. Roy.
  Statist. Soc., Ser. B}, \textbf{77}, 475--507.

\bibitem[{Cs\"{o}rg\H{o} and Horv\'{a}th(1997)}]{csorgo}
Cs\"{o}rg\H{o}, M. and Horv\'{a}th, L. (1997) \emph{Limit Theorems in
  Change-Point Analysis}. John Wiley and Sons, New York.

\bibitem[{Cule, Samworth and Stewart(2010)}]{cule2010maximum}
Cule, M., Samworth, R. and Stewart, M. (2010) Maximum likelihood estimation of
  a multi-dimensional log-concave density. \emph{J. Roy. Statist. Soc., Ser. B.
  (with discussion)}, \textbf{72}, 545--607.

\bibitem[{D{\"u}mbgen and Rufibach(2009)}]{dumbgen2009maximum}
D{\"u}mbgen, L. and Rufibach, K. (2009) Maximum likelihood estimation of a
  log-concave density and its distribution function: Basic properties and
  uniform consistency. \emph{Bernoulli}, \textbf{15}, 40--68.

\bibitem[{Enikeeva and Harchaoui(2019)}]{enikeeva2013highdimensional}
Enikeeva, F. and Harchaoui, Z. (2019) High-dimensional change-point detection
  under sparse alternatives. \emph{Ann. Statist.}, \textbf{47}, 2051--2079.

\bibitem[{Fryzlewicz(2014)}]{fryzlewicz2014wild}
Fryzlewicz, P. (2014) Wild binary segmentation for multiple change-point
  detection. \emph{Ann. Statist.}, \textbf{42}, 2243--2281.

\bibitem[{Hampel(1974)}]{hampel}
Hampel, F.~R. (1974) The influence curve and its role in robust estimation.
  \emph{J. Amer. Statist. Assoc.}, \textbf{69}, 383--393.

\bibitem[{Henry, Simani and Patton(2010)}]{henry}
Henry, D., Simani, S. and Patton, R. (2010) Fault detection and diagnosis for
  aeronautic and aerospace missions. \emph{Fault Tolerant Flight Control},
  \textbf{399}, 91--128.

\bibitem[{Horv{\'a}th and Hu{\v{s}}kov{\'a}(2012)}]{HorvathHuskova2012}
Horv{\'a}th, L. and Hu{\v{s}}kov{\'a}, M. (2012) Change-point detection in
  panel data. \emph{J. Time Series Anal.}, \textbf{33}, 631--648.

\bibitem[{Horv{\'a}th and Rice(2014)}]{rice}
Horv{\'a}th, L. and Rice, G. (2014) Extensions of some classical methods in
  change point analysis. \emph{TEST}, \textbf{23}, 219--255.

\bibitem[{Huopaniemi et~al.(2014)Huopaniemi, Nadkarni, Nadukuru, Lotay, Ellis,
  Gottesman and Bottinger}]{huopaniemi2014disease}
Huopaniemi, I., Nadkarni, G., Nadukuru, R., Lotay, V., Ellis, S., Gottesman, O.
  and Bottinger, E.~P. (2014) Disease progression subtype discovery from
  longitudinal EMR data with a majority of missing values and unknown initial
  time points. In \emph{AMIA Annual Symposium Proceedings}, vol. 2014, 709,
  American Medical Informatics Association.

\bibitem[{Jirak(2015)}]{Jirak2015}
Jirak, M. (2015) Uniform change point tests in high dimension. \emph{Ann.
  Statist.}, \textbf{43}, 2451--2483.

\bibitem[{Kov{\'a}cs et~al.(2020)Kov{\'a}cs, Li, Haubner, Munk and
  B{\"u}hlmann}]{kovacs2020optimistic}
Kov{\'a}cs, S., Li, H., Haubner, L., Munk, A. and B{\"u}hlmann, P. (2020)
  Optimistic search strategy: Change point detection for large-scale data via
  adaptive logarithmic queries. \emph{arXiv preprint}, arxiv:2010.10194.

\bibitem[{Liu, Gao and Samworth(2021)}]{liu2021minimax}
Liu, H., Gao, C. and Samworth, R.~J. (2021) Minimax rates in sparse,
  high-dimensional change point detection. \emph{Ann. Statist.}, \textbf{49},
  1081--1112.

\bibitem[{Londschien, Kov{\'a}cs and B{\"u}hlmann(2021)}]{londschien2021change}
Londschien, M., Kov{\'a}cs, S. and B{\"u}hlmann, P. (2021) Change-Point
  Detection for Graphical Models in the Presence of Missing Values. \emph{J.
  Comput. Graph. Statist.}, to appear.

\bibitem[{Mazumder, Hastie and Tibshirani(2010)}]{mazumder2010spectral}
Mazumder, R., Hastie, T. and Tibshirani, R. (2010) Spectral regularization
  algorithms for learning large incomplete matrices. \emph{J. Mach. Learn.
  Res.}, \textbf{11}, 2287--2322.

\bibitem[{Mei(2010)}]{mei2010efficient}
Mei, Y. (2010) Efficient scalable schemes for monitoring a large number of data
  streams. \emph{Biometrika}, \textbf{97}, 419--433.

\bibitem[{Olshen et~al.(2004)Olshen, Venkatraman, Lucito and Wigler}]{olshen}
Olshen, A., Venkatraman, E., Lucito, R. and Wigler, M. (2004) Circular binary
  segmentation for the analysis of array-based DNA copy number data.
  \emph{Biostatistics}, \textbf{5}, 557--572.

\bibitem[{Padilla et~al.(2019)Padilla, Yu, Wang and
  Rinaldo}]{padilla2019optimal}
Padilla, O. H.~M., Yu, Y., Wang, D. and Rinaldo, A. (2019) Optimal
  nonparametric multivariate change point detection and localization.
  \emph{arXiv preprint}, arxiv:1910.13289.

\bibitem[{Page(1955)}]{Page1955}
Page, E.~S. (1955) A test for a change in a parameter occurring at an unknown
  point. \emph{Biometrika}, \textbf{42}, 523--527.

\bibitem[{Peng, Leckie and Ramamohanarao(2004)}]{peng2004proactively}
Peng, T., Leckie, C. and Ramamohanarao, K. (2004) Proactively detecting
  distributed denial of service attacks using source IP address monitoring. In
  \emph{International Conference on Research in Networking}, 771--782,
  Springer.

\bibitem[{Poore et~al.(2006)Poore, Samworth, White, Jones and
  McCave}]{poore2006neogene}
Poore, H., Samworth, R., White, N., Jones, S. and McCave, I. (2006) Neogene
  overflow of northern component water at the Greenland-Scotland Ridge.
  \emph{Geochem. Geophys. Geosyst.}, \textbf{7}, Q06010.

\bibitem[{Samworth and Poore(2005)}]{samworth2005understanding}
Samworth, R. and Poore, H. (2005) Understanding past ocean circulations: a
  nonparametric regression case study. \emph{Stat. Model.}, \textbf{5},
  289--307.

\bibitem[{Soh and Chandrasekaran(2017)}]{soh2017high}
Soh, Y.~S. and Chandrasekaran, V. (2017) High-dimensional change-point
  estimation: Combining filtering with convex optimization. \emph{Appl. Comput.
  Harmon. Anal.}, \textbf{43}, 122--147.

\bibitem[{Sparks, Keighley and Muscatello(2010)}]{sparks}
Sparks, R., Keighley, T. and Muscatello, D. (2010) Early warning CUSUM plans
  for surveillance of negative binomial daily disease counts. \emph{J. Appl.
  Statist.}, \textbf{37}, 1911--1929.

\bibitem[{Stewart and Sun(1990)}]{StewartSun1990}
Stewart, G.~W. and Sun, J. (1990) \emph{Matrix Perturbation Theory}. Academic
  Press, San Diego.

\bibitem[{Wainwright(2019)}]{Wainwright2019}
Wainwright, M.~J. (2019) \emph{High-Dimensional Statistics: A Non-Asymptotic
  Viewpoint}. Cambridge University Press, Cambridge.

\bibitem[{Wang(2016)}]{Wang2016}
Wang, T. (2016) \emph{Spectral Methods and Computational Trade-offs in
  High-dimensional Statistical Inference}. PhD Thesis, University of Cambridge.

\bibitem[{Wang and Samworth(2018)}]{wang2016highdimensional}
Wang, T. and Samworth, R.~J. (2018) High dimensional change point estimation
  via sparse projection. \emph{J. Roy. Statist. Soc., Ser. B}, \textbf{80},
  57--83.

\bibitem[{Wright and Miller(1996)}]{wright1996control}
Wright, J.~D. and Miller, K.~G. (1996) Control of North Atlantic Deep Water
  circulation by the Greenland-Scotland Ridge. \emph{Paleoceanography},
  \textbf{11}, 157--170.

\bibitem[{Xie, Huang and Willett(2013)}]{Xie_2013}
Xie, Y., Huang, J. and Willett, R. (2013) Change-Point Detection for
  High-Dimensional Time Series With Missing Data. \emph{IEEE Journal of
  Selected Topics in Signal Processing}, \textbf{7}, 12--27.

\bibitem[{Xie and Siegmund(2013)}]{xie2013sequential}
Xie, Y. and Siegmund, D. (2013) Sequential multi-sensor change-point detection.
  \emph{Ann. Statist.}, \textbf{41}, 670--692.

\bibitem[{Yu(1997)}]{Yu1997}
Yu, B. (1997) Assouad, Fano and Le Cam. In Pollard, D., Torgersen, E. and Yang
  G. L. (Eds.). \emph{Festschrift for Lucien Le Cam: Research Papers in
  Probability and Statistics}, 423--435.

\bibitem[{Yu, Wang and Samworth(2015)}]{YuWangSamworth2015}
Yu, Y., Wang, T. and Samworth, R.~J. (2015) A useful variant of the
  Davis--Kahan theorem for statisticians. \emph{Biometrika}, \textbf{102},
  315--323.

\bibitem[{Zhang et~al.(2010)Zhang, Siegmund, Ji and Li}]{Zhangetal2010}
Zhang, N.~R., Siegmund, D.~O., Ji, H. and Li, J.~Z. (2010) Detecting
  simultaneous changepoints in multiple sequences. \emph{Biometrika},
  \textbf{97}, 631--645.

\bibitem[{Zhu, Wang and Samworth(2019)}]{zhu2019high}
Zhu, Z., Wang, T. and Samworth, R.~J. (2019) High-dimensional principal
  component analysis with heterogeneous missingness. \emph{arXiv preprint},
  arxiv:1906.12125.

\end{thebibliography}
\end{document}